\documentclass[conference]{IEEEtran}
\usepackage{graphicx}
\usepackage{epstopdf}

\usepackage[cmex10]{amsmath}
\usepackage{algorithmic}
\usepackage{algorithm}
\usepackage{amssymb}
\usepackage{cite}
\usepackage{amsthm} 
\usepackage{xcolor}

\usepackage{multirow}
\usepackage{url}
\usepackage{bm}
\usepackage{mathtools}
\usepackage{multirow}
\usepackage{pbox}
\usepackage{verbatim}
\usepackage[justification=centering]{caption}

\newtheorem{definition}{\textbf{Definition}}\newtheorem{lemma}{\textbf{Lemma}}\newtheorem{theorem}{\textbf{Theorem}}\newtheorem{corollary}{Corollary}\newtheorem{remark}{Remark}

\newcommand{\RN}[1]{%
  \textup{\uppercase\expandafter{\romannumeral#1}}%
}
\newcommand{\q}[1]{``#1''}
\pdfoutput=1

\begin{document}

\title{
On Shared Rate Time Series for Mobile Users\\
in Poisson Networks}

\author{\IEEEauthorblockN{Pranav Madadi, Fran\c cois Baccelli and Gustavo de Veciana}
\thanks{This paper was presented in part at Wi-Opt 2016.}

}

\maketitle

\thispagestyle{plain}
\pagestyle{plain}

\begin{abstract}

This paper focuses on modeling and analysis of the
temporal performance variations experienced by a mobile user
in a wireless network and its impact on system-level design.
We consider a simple stochastic geometry model: 
the infrastructure nodes are Poisson distributed while the user's
motion is the simplest possible i.e., constant velocity on a straight line. 
We first characterize variations in the SNR process, and associated downlink Shannon
rate, resulting from variations in the infrastructure geometry seen by the mobile.
Specifically, by making a connection between stochastic geometry and queueing theory
the level crossings of the SNR process are shown to form 
an alternating renewal process whose distribution can be completely characterized.
For large/small SNR levels, and associated rare events, we further derive 
simple distributional (exponential) models. 
We then characterize the second major contributor to variation,
associated with changes in the number of other users sharing infrastructure.
Combining these two effects, we study what are the dominant factors 
(infrastructure geometry or sharing number) given mobile experiences 
a very high/low shared rate.These results are then used to evaluate and optimize the system-level Quality of Service (QoS) and system-level capacity to support
mobile users sharing wireless infrastructure; including 
mobile devices streaming video which proactively buffer
content to prevent rebuffering and mobiles which are downloading large files. 
Finally, we use simulation to assess the fidelity of this model 
and its robustness to factors which are presently  
taken into account. 


\end{abstract}

\begin{IEEEkeywords}
Poisson networks, wireless, Signal-to-Noise Ratio (SNR), Shannon rate, mobile user, temporal variations.\end{IEEEkeywords}

\section{Introduction}


The primary aim of this paper is to model and study 
the temporal capacity variations experienced by wireless users 
moving through space.
These are driven by geometric variations in the spatial and environmental 
relationships (associations) to the infrastructure, the
presence of other users, as well as random fluctuations intrinsic
to wireless channels, e.g., fast fading.
Our focus on temporal variability should be 
contrasted with the extensive work characterizing 
spatial variability as seen by randomly located users, e.g., through
metrics such as coverage probability, spatial density of throughput, 
90\% quantile rate, \q{edge} capacity, spectral efficiency, etc. 
Although space and time may be related through averages (when
ergodicity holds), the temporal characteristics of the stochastic processes
modeling  a mobile's capacity variations are relatively 
unexplored beyond correlation analysis and of great practical interest.


Indeed, an increasing volume of data traffic is generated by 
wireless devices while moving, 
e.g., 20-30\% of cellular data is generated during 
commute hours, \cite{Pranav2}. 
In the future, with increasing use of public transportation and/or
the emergence of driverless cars,
this volume could grow substantially. 
The Quality of Service/Experience (QoS/E) seen by such users 
can be highly dependent on the capacity variations 
they see as they move through space. 
This is particularly the case for applications that operate
over longer time scales, e.g., video/audio streaming, navigation/augmented reality, 
and real-time services, which may not be able to smooth substantial capacity 
variation through buffering, but also for the opportunistic 
transfers of large files over heterogeneous networks, e.g., WIFI offloading. 
More broadly, increases in wireless network capacity have 
been achieved through heterogeneous network densification leveraging 
technologies providing different coverage-throughput
trade offs, e.g., cellular (macro, pico, femto cells),
WIFI and perhaps, in the future, mmWave access points.
Characterizing and managing the capacity variations mobile users 
would see across such networks is a challenging but important
problem towards understanding the efficiency and performance
offloading/onloading based services. 

 
{\em Related work.}
There is a rich literature on modeling
spatial capacity variability in wireless infrastructure
for a randomly located users.  Of particular relevance
is that based on stochastic geometry,
which captures the effect of the variability
in base station locations, as well as the variability in the environment
through shadowing, and in the channel through fading, 
see e.g. \cite{BB} for a survey on the matter. 
By contrast work studying the temporal capacity variations for a 
user on the move in a cellular network is limited. 

There is also significant related work on Delay-Tolerant Networks (DTN). 
This literature considers mobile nodes, where the contact 
duration and the inter-contact time are defined and empirically 
measured from real traces \cite{20}, \cite{13} as well as through 
mobility models \cite{10}. In addition to mobility 
induced inter-contact processes, \cite{1} considers other factors 
like user availability. However, this work is in the context of 
opportunistic ad-hoc communication networks where a set of mobile 
nodes are moving under different mobility patterns. Our focus is 
on studying the {\em continuous-parameter stochastic process} 
experienced by a tagged mobile user traversing a static pattern 
of nodes modeling the wireless infrastructure.


There certainly is a lot of interest in studying how to design
networks to better address the needs of mobile users or to
leverage user mobility for the offloading/onloading traffic. 
For example \cite{iitk} studies  how user mobility patterns and 
users perceived QoS might drive the selection of macro-cell upgrades. 
The work in \cite{Hong} examines the effectiveness of algorithms for optimizing 
offloading to a set of spatially distributed WIFI APs.  
The work in \cite{ZhD14} evaluates how proactive knowledge
of capacity variations could be used in designing new models
for video delivery. These works exemplify applications
and engineering problems which depend critically on the temporal capacity
variability that mobile users would experience, but
do not directly address the characteristics of such processes.

{\em Key Questions.} 
Two primary sources of temporal variation in a mobile's capacity are 
the SNR, i.e., associated with changes in the users geometric relationship
to the infrastructure, and the sharing number, i.e., the number of 
other mobiles sharing the resource.
Our goal in this paper is to study the relative impacts these have on
mobile users' QoS/E.  In this setting several basic questions arise: 

\begin{enumerate}
\item 
\noindent{\bf \em Characterization of the SNR process.}
As a mobile user moves through a wireless network infrastructure 
associating with the closest node, its SNR process and thus associated
peak rate (i.e., without accounting for sharing)
will experience peaks and valleys. What is the intensity
of peaks and valleys? Does it match the rate of cell boundary crossings?
Given an SNR threshold, can one characterize the temporal
characteristics of the on/off level crossing process
associated with being above 
and below the threshold, i.e., the coverage and outage durations?

\item 
\noindent{\bf \em Characterization of the sharing number process.}
Assuming a population of other users sharing the network, what
are the characteristics of the sharing number process seen by a mobile?
If the network is shared by heterogeneous users, i.e., static, pedestrian, 
and users on public transport and/or a road system, how will this 
bias what they see? 
\item 
\noindent{\bf \em Smoothness of the effective rate process.}
How smooth is the bit rate obtained by the mobile seen as a function
of time? How often does this rate incur discontinuities, 
trend changes, large jumps or other phenomena negatively
impacting real time applications?
\item 
\noindent{\bf \em Characterization of rare events.}
Conditional on a rare event, i.e., very poor or very good user rate, 
what is the relative contribution of the user's location vs network congestion?  
Can one characterize the time scales for rare
event occurrences?
\item 
\noindent{\bf \em Applications to QoE.}  
What are the implications of the temporal capacity variations mobiles'
see on their application-level QoE and system-level performance, e.g., 
acceptable density of video streaming users, or download delays of 
large files? 
\end{enumerate}

To the best of our knowledge the analysis of the SNR and shared rate processes 
for Poisson wireless infrastructure developed in this paper is new and provides a first order 
answer to these basic questions. While our models are simplified
they should be viewed as a first, and necessary, step towards 
characterizing temporal variability in more complex random structures,
e.g., SINR variations which in this paper are only explored 
via simulation for comparison to SNR process characteristics.
Generalizations to SINR processes would be desirable, particularly
for systems operating in the interference limited regime. 
Such extensions may be tractable based on results known for
shot noise fields (see \cite{BB}, Volume 1), 
yet are beyond the scope of this paper.

{\em Contributions and Organization.} 
Section \RN{2} introduces the basic model studied in this paper. 
We consider a tagged user moving at a fixed velocity along a straight line 
through a shared wireless network such that:
(a) the access points/base stations are a realization of a Poisson point process; 
(b) other users sharing the network also form a Poisson (or possibly Cox) point process;
(c) all users associate with the closest base station; 
(d) resources are shared equally with other users sharing the network;
(e) a standard distance based path loss is used with neither shadowing nor fading.
Our goal is to characterize the \textit{shared rate} process seen by the tagged user  
by studying the \textit{SNR} and \textit{sharing number} processes.

In Section \RN{3} we present our results on the SNR process, 
including the intensity of peaks and valleys and, 
given a SNR threshold, we provide a a complete characterization 
of the on/off level crossing process as an alternating-renewal process.
Equivalently, this characterizes the durations for coverage/outage events seen by the mobile. 
Interestingly this is achieved by establishing a connection between 
the time-varying geometry seen by the mobile user and an 
associated queueing process.  We then provide asymptotics for the 
likelihood of high and low SNR, and show that, after appropriate rescaling,
the time intervals between up crossings are exponential. Which verifies
the \q{rarity hence exponentiality} principle \cite{Kei,Ald} for such
events.  

Section \RN{4} provides a detailed discussion of the
sharing number process,  i.e., the number of other users which 
are covered, i.e., meet an SNR threshold, and share 
infrastructure nodes with the tagged user. The process is
somewhat complex and not independent of the SNR process;
so we introduce a simpler bounding process which is
used in Section \RN{5} to characterize the {\em shared
rate process} seen by our mobile. 
In particular we show that rare events associated
with high shared rates, are associated
with high SNR (i.e., mobile's proximity to base stations) with
no other users sharing the resources. Similarly low shared
rates, arise when mobiles are far from the base station, and there
is a number of sharing users inversely proportional to
the low rate in question. 
We also provide an asymptotic characterization for the re-scaled
interarrival times associated with upcrossings.

In Section \RN{6} we present a discussion of discrete event simulation results 
which are used to assess the robustness of this model to perturbations that were 
not yet taken into account in the analysis.  For example the relative impact of 
variability associated with the changing geometry (proximity of base stations) 
seen by a mobile versus that associated with channel variability due to channel fades.



In Section \RN{7} we leverage our results to evaluate the QoS experienced 
by mobile users in two concrete scenarios. First we consider
the delivery of streaming video to mobile users which are
able to store future video frames to prevent rebuffering. 
The primary question addressed is about the maximal density 
of such users which can be supported while ensuring that in the long term
no rebuffering is required. 
Our second application considers the distribution for the 
delays experienced by a mobile user attempting to download a large file. 
These two examples give a system-level view on the performance that mobile
users would see when downloading large files opportunistically via some
WIFI infrastructure.

Section \RN{8} briefly discusses some important extensions 
of our results to the case of heterogeneous wireless infrastructures
and the case where the locations of mobile users follow
a Cox process associated with a road network. The results show
how heterogeneous technologies impact the temporal variations
in the mobile users SNR process. 
The results also show how a mobile user
on a roadway is likely to see poorer performance than
a stationary or pedestrian user.
Section \RN{9} concludes the paper. 


\section{System Model}

Consider an infrastructure based wireless network consisting of nodes e.g., base stations/Wifi hotspots, denoted through locations on the Euclidean plane. The configuration of the nodes is assumed to be a realization of Poisson point process $\Phi = \{X_1,X_2, ..\} $ in $\mathbb{R}^2$ with intensity $\lambda$. We consider a tagged user moving at a fixed velocity $v$ along a straight line starting from the origin at time $t=0$. The mobile user shares the network with other (static) users spatially distributed according to another independent Poisson point process of intensity $\xi$.

All users associate with the closest infrastructure node. For each $X_i \in \Phi$ one can define a set of locations which are closer to $X_i$ than any other point in $\Phi \setminus {X_i}$. This is a convex polygon known as the Voronoi cell associated with $X_i$ \cite{BB}. The collection of such cells forms a tessellation of the plane called the Voronoi tessellation see e.g., Fig.\ref{fig1}. Thus, all the users located within the Voronoi cell of node $X_i$ associate with $X_i$.


Let $(X(t) ,t \geq 0)$, denote the random process where, $X(t) \in \Phi$ is the closest node to the mobile user at time $t$. 
Let $(L(t) ,t \geq 0)$
be the process denoting the distance from
the mobile user to its closest node.

\begin{figure}[!t]
\centering
\includegraphics[scale =0.5]{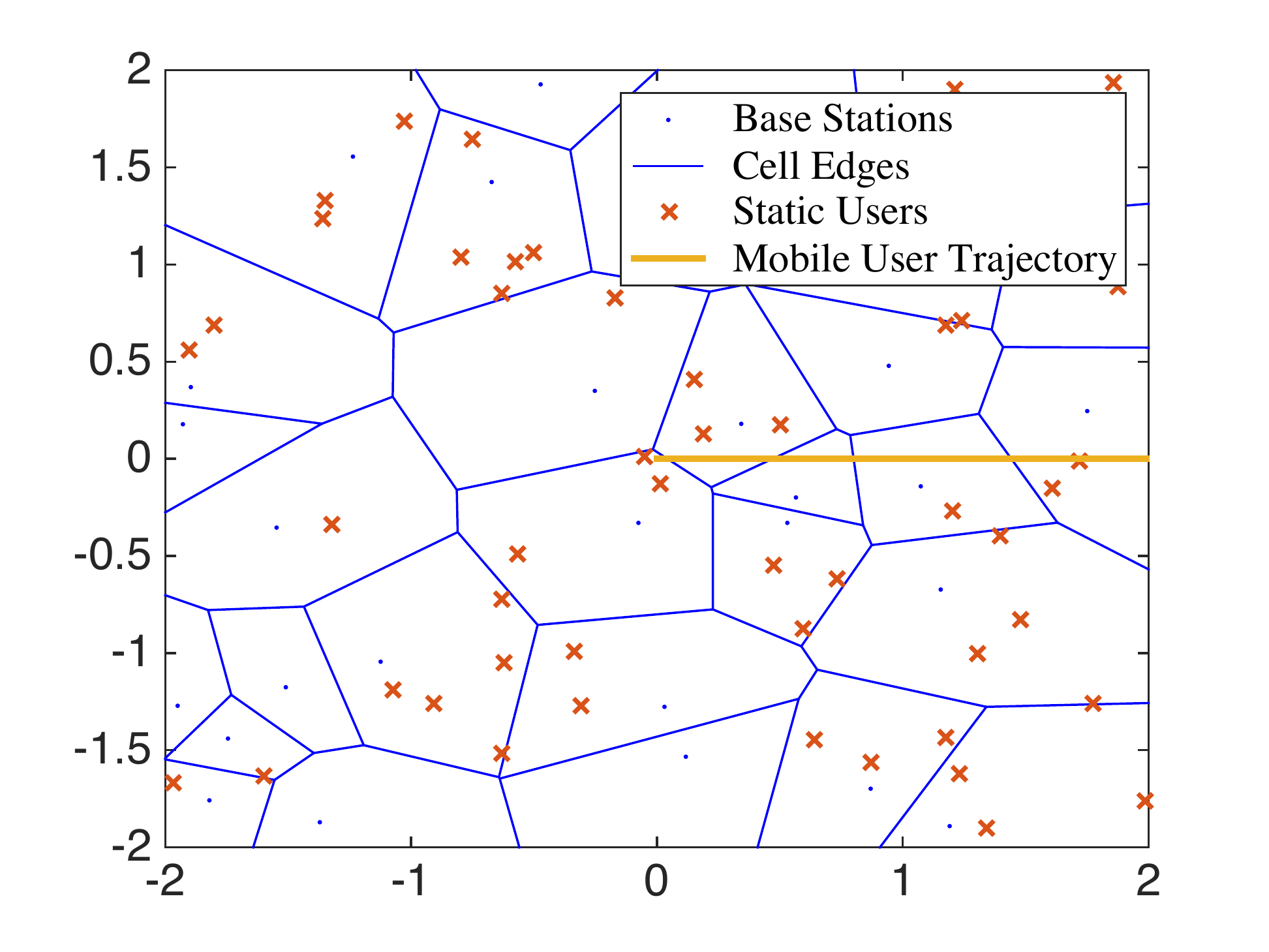}

\caption{Mobile user motion in a sample of the Poisson cellular network.}
\label{fig1}
\end{figure}

We consider downlink transmissions and assume that all nodes transmit at a fixed power $p$. Based on the classical power law path loss model and in the absence of fading and shadowing, the \textit{Signal-to-Noise Ratio} (SNR) \textit{process} of the mobile user is
\begin{equation}
\label{eq:snr}
\mbox{SNR(t)} = \frac{p L(t)^ {-\beta}}{w}, ~~ \mbox{for} ~ ~\beta > 2, ~~ t \geq 0,
\end{equation}
where $w$ denotes the noise power.

We will assume that the users are only \emph{served} if their SNR exceeds a given threshold $\gamma$. It follows from \eqref{eq:snr} that a user is served if it is within a distance $r_{\gamma} = (\frac{p}{w \gamma})^{1/\beta}$ from the closest node. We refer to $r_{\gamma}$ as the \textit{radius of coverage}. Note that a user may hence associate with a node but not be served.

The \textit{Shannon rate process}, $(R^{(\gamma)}(t), t \geq 0)$, seen by the mobile user is directly determined by the SNR process through the relation:
\begin{equation}
\label{eq:shannon}
R^{(\gamma)}(t) = 
\begin{cases}
a \log\left(1 + {\mathrm {SNR}}(t)\right) & \text{if}
~~ L(t) \leq r_{\gamma}, \\
0~~ &\text{otherwise},
\end{cases}
\end{equation}
where $a$ is a constant depending on the available bandwidth.

We assume that each node shares its resources equally among the users
it serves (e.g. through some time sharing scheme).
We define the  \textit{sharing number process} $(N^{(\gamma)}(t), t \geq 0)$,
where $N^{(\gamma)}(t)$ is the number of static users the node
associated with the mobile user serves provided that it is itself served,
and 0 if the mobile user is not served.
In other words, if the mobile user is served by its closest node
$X(t)$ at time $t$, it shares the resources with $N^{(\gamma)}(t)$
static users. Directly determined by the sharing number, the
\textit{sharing factor} $(F^{(\gamma)}(t), t \geq 0)$, is defined by:
\begin{equation}
F^{(\gamma)}(t) = \frac{1}{1 + N^{(\gamma)}(t)}.
\end{equation}

Finally, the \textit{shared rate process} seen by the mobile
user $(S^{(\gamma)}(t), t \geq 0)$, is given by
\begin{equation}
\label{eq:rate}
S^{(\gamma)}(t) = R^{(\gamma)}(t) \times F^{(\gamma)}(t).
\end{equation}

Our aim is to characterize this process. In the next two
sections, we first study the two underlying processes namely:
(1) the SNR process $(\mbox{SNR}(t), t \geq 0) $ and its level
sets, and (2) the Sharing number process $(N^{(\gamma)}(t), t \geq 0 )$. 

In the sequel, for any stationary random process say for e.g.,
$(S^{(\gamma)}(t) ,t \geq 0)$, $S^{(\gamma)}$ represents a random
variable with distribution equal to the stationary distribution of this process.

A summary of the key notation is provided in the Table \RN{1}.




\begin{table*}
\label{table:def}
\normalsize
\begin{center}
\begin{tabular}{ |c|c| } 
\hline
\textbf{Symbol} & \textbf{Definition} \\
\hline
$\Phi , \lambda $ & Poisson point process of nodes and its intensity  \\
\hline
$\xi$ & Intensity of mobile users\\
\hline
$p$ & Constant power transmitted by nodes \\
\hline
$(X(t) ,t > 0)$ & Random process denoting the closest node to the mobile \\ 
\hline
$(L(t) ,t > 0)$ & Random process denoting the distance to the closest node \\ 
\hline
$(\mbox{SNR}(t),t > 0)$ & Signal-to-Noise ratio random process\\ 
\hline
$(C^{(\gamma)}(t),t > 0)$ &  Signal-to-Noise ratio level crossing process\\
\hline
$(R^{(\gamma)}(t),t > 0)$ &  Shannon rate random process\\
\hline
$(N^{(\gamma)}(t),t > 0)$ & Random process denoting the number of users sharing \\ 
\hline
$(S^{(\gamma)}(t),t > 0) $ &  Shared rate random process\\ 
\hline
$\gamma$ & Threshold on signal-to-noise ratio\\
\hline
$r_{\gamma}$ & Radius of coverage for threshold $\gamma$\\
\hline
\end{tabular}
\end{center}
\caption{Table of Notation.}

\end{table*}


\section{Characterization of the SNR process}

This section is structured as follows:
we start by considering the on-off coverage structure of the SNR process and
then study in more detail the characteristics of its fluctuations.
We then analyze the scale of the inter-occurrence times of
certain rare events.

\subsection{Analysis of the SNR Level Crossing Process}

A first order question is whether the mobile is covered or not. To that end we define the SNR level crossing process as follows:
\begin{definition}
Given an SNR threshold $\gamma$, the SNR level crossing process $(C^{(\gamma)}(t), t \geq 0 )$ is defined as $C^{(\gamma)}(t) = \bm{1}(\mbox{SNR}(t) \geq \gamma)$ as shown in Fig \ref{on-off process}.
\end{definition}
Clearly this is an on-off process, where the on and off periods correspond to the coverage and outage periods respectively. The process alternates between \q{on} intervals of length $(B^{(\gamma)}_n, n \geq 1) $ and \q{off} intervals of length $( I^{(\gamma)}_n, n \geq 1) $ as depicted in Fig.~\ref{on-off process}. We also define the sequence of SNR up-crossing times $(T^{(\gamma)}_n , n \geq 1)$ for a given threshold $\gamma$. Let $ V^{(\gamma)} \sim T^{(\gamma)}_n -T^{(\gamma)}_{n-1}$ be a random variable whose distribution is that associated with up-crossing inter arrivals as illustrated in Fig. \ref{on-off process}.

\begin{figure}[!t]
\centering
\includegraphics[scale=0.38]{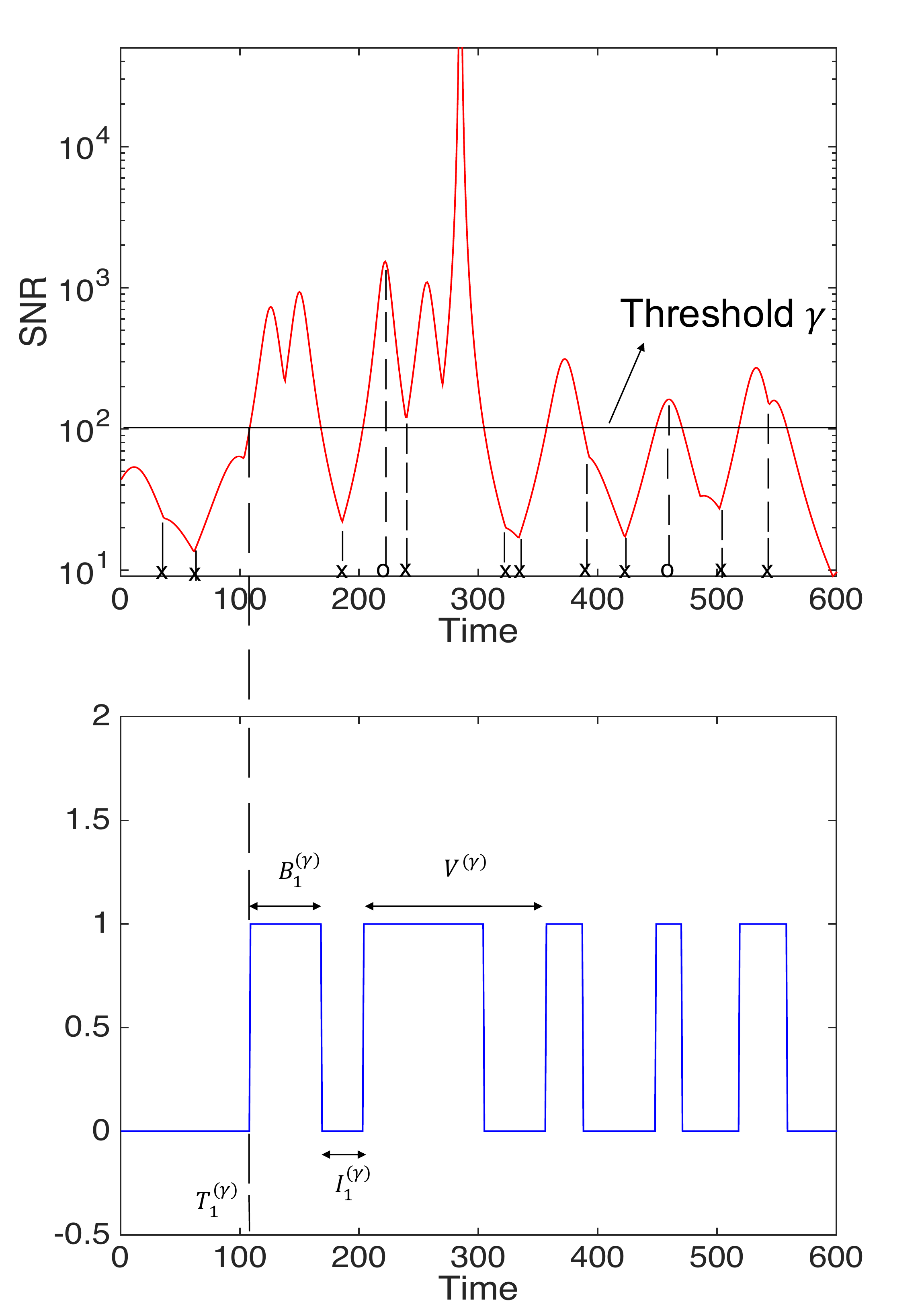}

\caption{Level crossings of the SNR process as an On-Off process (bottom) and
the point process of its local maxima $\chi$, 
denoted by ``o'' on the $x$ axis, together with edge crossings, denoted
by ``$\times$'' (top).}
\label{on-off process}
\end{figure}

In order to characterize the SNR level crossing process,
we establish a connection between the
time-varying geometry seen by the mobile user and an
associated queueing process.

Let $D^{(\gamma)}(t)$ denote the closed disc of radius $r_{\gamma}$ centered on the mobile user's location at time $t$. This closed disc follows the mobile user's motion along the straight line. 
Let $K^{(\gamma)}(t) = |\Phi \cap D^{(\gamma)}(t)|$ denote the number of nodes in the disc.

The following theorem provides a simple characterization
of $(K^{(\gamma)}(t), t \geq 0 )$ which in turn will help studying
the SNR level crossing process.

\begin{theorem}
\label{queueing model}
The process $ (K^{(\gamma)}(t), t \geq 0 )$ is equivalent to that modeling the number of customers in an $M/GI/\infty $ queue with arrival rate $\lambda^{({\gamma})} = 2r_{\gamma}v \lambda$ and i.i.d. service times with density
\end{theorem}
\begin{equation}
\label{eq:service density}
f_{W^{(\gamma)}}(s) = 
\begin{cases}
\frac{v^2s}{2r_{\gamma}\sqrt{4r_{\gamma}^2 - v^2s^2}} ~ ~ & \text{for} ~ s \in [0,2r_{\gamma}/v],\\
0 ~~ & \text{otherwise}.
\end{cases}
\end{equation}

\begin{proof}
The entry of a node into the closed disc $D^{(\gamma)}(t)$ can be viewed as an arrival to the queue. The amount of time spent by the node in $D^{(\gamma)}(t)$ corresponds to its service time and thus the exit from $D^{(\gamma)}(t)$ its departure from the queue. 

We first show that the arrival process to the disc and thus the queue, is Poisson. We begin by proving that the arrival process has independent and stationary increments.The probability that there is an arrival in the next $\epsilon$ seconds is the probability that there is a node in the area 
$A^{(\gamma)}_{\epsilon}={2r_{\gamma}v \epsilon}$ as depicted in Fig. \ref{disc}. Since nodes are distributed according to a homogeneous Poisson point process of intensity $\lambda$, the number of nodes in any closed set of area $b$ follows the Poisson distribution with parameter $\lambda b$. Thus, for any $\epsilon > 0$, the increments in the arrival process have the same distribution. Also, the number of nodes in any two disjoint closed sets are independent. Thus, the arrival process has independent stationary Poisson increments.
For a small value of $\epsilon$, the probability that there is a single arrival is given by $\lambda A^{(\gamma)}_{\epsilon} + o(\epsilon) $.  Thus, the arrival process is Poisson with rate $2r_{\gamma}v \lambda$.
%
%

Every node that enters the moving closed ball stays in it for a time that depends on its entry locations and is proportional to the chord length as shown in Fig. \ref{disc}. This corresponds to its service time in the queue.

\begin{figure}[!t]
\centering
\includegraphics[scale = 0.26]{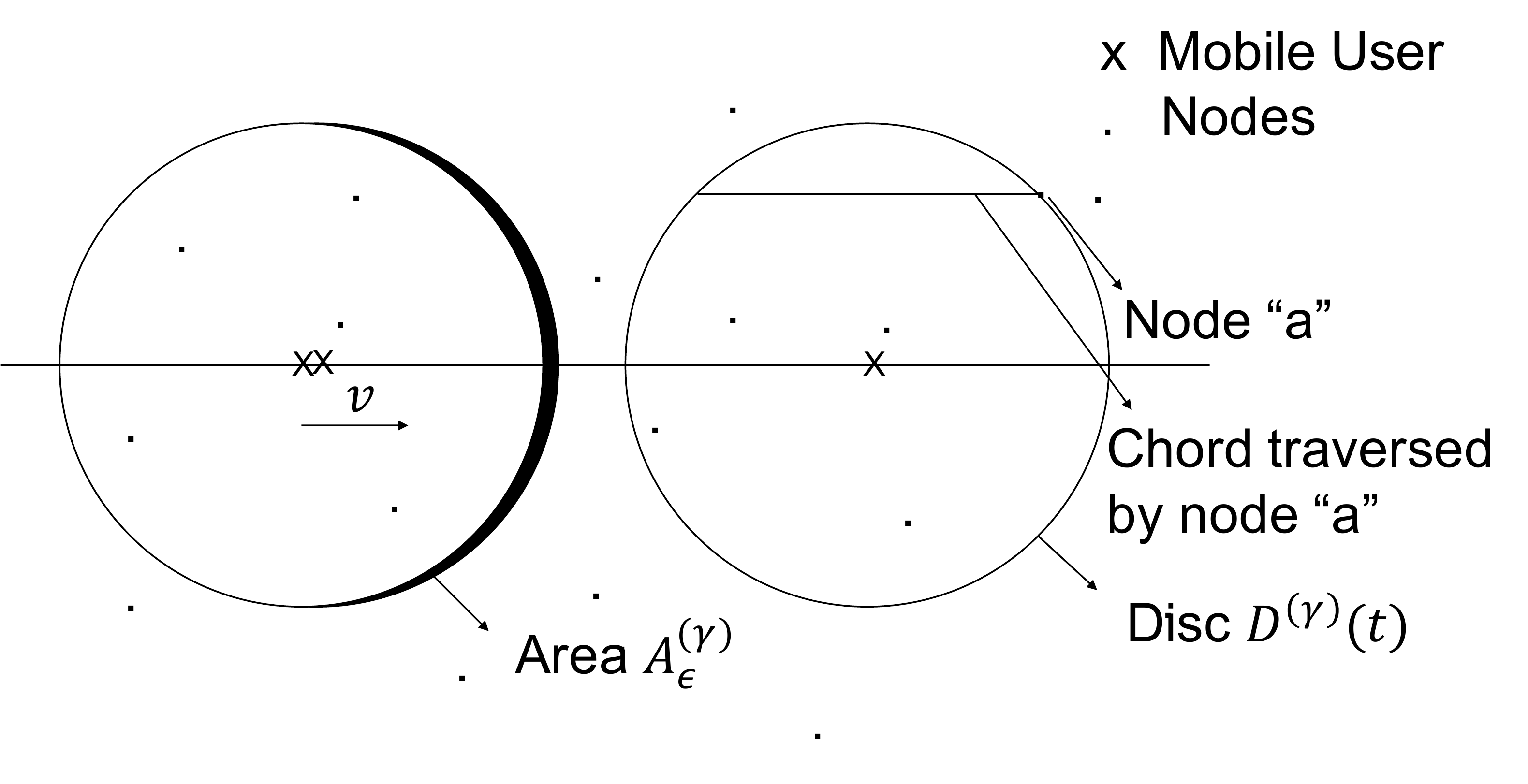}
\caption{The disc around the mobile user with area  $A^{(\gamma)}_{\epsilon}$  and the chord along which the node moves.}
\label{disc}
\end{figure}

Since the mobile moves at constant velocity, the distribution of the service times can be derived from the distribution of the chord lengths. Without loss of generality suppose the mobile moves along the x-axis, then clearly the $y$ coordinate of a typical node entering the disc,  $Y^{(\gamma)}$ is uniform on [$-r_{\gamma} , r_{\gamma}$]. The random variable $W^{(\gamma)} $ representing the service time is then given by:
$$
W^{(\gamma)} = \frac{2 \sqrt{r_{\gamma}^2 - (Y^{(\gamma)})^2}}{v}.
$$
The density of $W^{(\gamma)}$ is \eqref{eq:service density}, and the mean service time is $E[W^{(\gamma)}] = \frac{\pi r_{\gamma}}{2v}$.

Thus, the process $(K^{(\gamma)}(t), t \geq 0)$ capturing the number of nodes in the moving disc follows the dynamics of the number of customers in an $M/GI/\infty$ queue with arrival rate $\lambda^{({\gamma})}= 2r_{\gamma} v \lambda$ and general independent service times following the distribution of $W^{(\gamma)}$. It follows that the stationary distribution for $K^{(\gamma)}(t)$ is Poisson with mean $\pi r_{\gamma}^2 \lambda$ .
\end{proof}

Given the connection to an $M/GI/\infty$ queueing model, the SNR level crossing process is an alternating renewal process defined as follows:
\begin{definition}
\label{alternating renewal process}
A process alternating between successive on and off intervals is an alternating renewal process if the sequences of on period
$(B^{(\gamma)}_n , n\geq 1) $ and off period $(I^{(\gamma)}_n , n\geq 1)$ are independent sequences of i.i.d. non-negative random variables.
\end{definition}

\begin{theorem}
\label{level crossing characterization}
For all $\gamma>0$, the SNR level crossing process, $(C^{(\gamma)}(t), t \geq0)$, is an alternating-renewal process. Further, its typical on period, $B^{(\gamma)}$, and off period, $I^{(\gamma)}$, are distributed as the busy and idle periods of an $M/GI/\infty$ queue with arrival rate $\lambda^{({\gamma})} = 2r_{\gamma}v \lambda$ and i.i.d. service times with distribution given in \eqref{eq:service density}.
Thus, $I^{(\gamma)} \sim \exp(2 \lambda v r_{\gamma})$ and the busy period distribution can be explicitly characterized as in ~\cite{Makowski}. Also, in the stationary regime, the probability that the SNR level crossing process is \q{on} is $1 - e^{-\lambda\pi r_{\gamma}^2}$.
\end{theorem}

\begin{proof}
Let us consider the $M/GI/\infty$ queue defined in Theorem \ref{queueing model}. The queue being empty implies that there are no nodes closer than distance $r_{\gamma}$ from the tagged user. Thus, an off period of the associated SNR level crossing process is equivalent to an idle period of the queue. Similarly, the busy period is equivalent to an on period. 

Since the arrival process is Poisson, the inter arrival times are exponential with parameter $\lambda^{(\gamma)}= 2 \lambda v r_{\gamma}$. Because of the memoryless property, all the idle periods obey the same distribution.

Let $\hat{B}^{(\gamma)}$ denote a random variable representing the forward recurrence time associated with the on time defined as
\begin{equation}
\label{eq:forward recurrence}
 \mathbb{P}(\hat{B}^{(\gamma)} > x) = \frac{1}{\mathbb{E}[B^{(\gamma)}]} \int_x^{\infty} \mathbb{P} (B^{(\gamma)} > z) dz.
\end{equation}
Using the correspondence of the SNR on times with the busy period of the $M/GI/\infty$ queue considered, its distribution is same as that given in Theorem \ref{makowski} in the Appendix with service time $S \sim W^{(\gamma)}$ and
\begin{equation}
\label{eq: rho with r}
\rho^{(\gamma)} = 2r_{\gamma}v \lambda \mathbb{E}[W^{(\gamma)}] = \lambda\pi r_{\gamma}^2 ~\mbox , ~ ~ \nu^{(\gamma)} = 1 - e^{-\rho^{(\gamma)}}.
\end{equation}

Also, the busy period $B^{(\gamma)}$ depends upon the arrivals and service times of customers arriving after the customer initiating the busy period which are independent of the past arrivals. Thus the busy period $B^{(\gamma)}$ and idle period $I^{(\gamma)}$ are independent. Hence, the SNR level crossing process form an alternating-renewal process.

The mean of the busy period of the queue is given by \cite{Makowski}:
\begin{equation}
\label{eq:mean busy}
\mathbb{E}[B^{({\gamma})}] = \frac{\nu^{({\gamma})}}{2\lambda vr_{\gamma} (1-\nu^{({\gamma})})}.
\end{equation}

Thus, the probability that the SNR level crossing process is \q{on} is $1 - e^{-\lambda\pi r_{\gamma}^2}$.

\end{proof}

It is easy to see that all path loss functions which are
monotonic lead to an analogue of Theorem
\ref{level crossing characterization}.

\subsection{Fluctuations of the SNR Process}

The SNR process is pathwise continuous and has continuous derivatives
except at a countable set of times
corresponding to Voronoi cell edge crossings where the mobile sees
a hand-off and where the first derivative of the SNR process
is discontinuous.

Since the mobile is moving with a fixed velocity $v$,
the intensity of the cell-edge crossings is given by
$4 v \sqrt{\lambda} / \pi$ \cite{Miles1988}, \cite{okabe}.


As the mobile traverses a particular cell, there is a time where it
is the closest to the associated node. The SNR process has a local maximum
at these times. 
Perhaps surprisingly such maxima may happen at the edge of the 
cell. Let us denote the point process of \textit{interior} maxima by
$\chi$. The (time) intensity of $\chi$ is given in the following theorem.


\begin{theorem}
\label{th:thm1}
The intensity of the point process $\chi$ of the SNR maxima occurring in the interior of cells is $ v \sqrt{\lambda} $.
\end{theorem}

Proof is given in Appendix A.

Thus, the fraction of SNR maxima that happen in the interior of the cell is given by $\pi/ 4$. In other words, $1-\frac{\pi}{4}$ of the SNR maxima seen by the mobile occur at the cell edges.

\subsection{{Rare Events}}


In this subsection we consider certain rare events such as the occurrence of large SNR values.
We show that the
\q{rarity hence exponentiality} principle \cite{Kei,Ald}
applies here. We also give the scales at which the inter-arrival
time of high SNR is close to an exponential.
As we shall see in Section \RN{7}, these asymptotic
results can also be used for moderate values of SNR for the parameters
typically found in wireless networks.

The following theorem gives a characterization of the asymptotic
behavior of the distribution of the random variable $V^{(\gamma)}$
corresponding to the up-crossings as $\gamma \to \infty$ and as
$\gamma \to 0$, which can be seen as \q{good} and \q{bad}
events respectively.

\begin{theorem}
\label{asymptotic}
For all $\gamma >0$, the up-crossings of the SNR level crossing process, $(T_{n}^{(\gamma)} , n \geq 1),$ constitute a renewal process. Let $V^{(\gamma)}$ be the typical inter arrival for this  process. Let $f(\gamma) =2 \lambda vr_{\gamma} $,
then 
$$ \lim_{\gamma \to \infty} f(\gamma) V^{(\gamma)}  \xrightarrow{d} \exp(1).  $$ 
Let $g(\gamma) = 2 \lambda vr_{\gamma} e^{-\lambda \pi r_{\gamma}^2}$ , then
$$ \lim_{\gamma \to 0} g(\gamma) V^{(\gamma)}  \xrightarrow{d} 
\exp(1).$$
\end{theorem}
\begin{proof}
The proof is given in the Appendix.\\
Notice that as $\gamma \to \infty$, the radius of coverage $r_{\gamma} \to 0$ and vice versa. The scale $f(\gamma)^{-1}$
of the inter-event times of \q{good events}  
is sub-linear in $r_{\gamma}$ when $r_{\gamma}$ tends to 0,
whereas the scale $g(\gamma)^{-1}$ for \q{bad events}
is exponential in the variable $r_{\gamma}^2$
when $r_{\gamma}$ tends to infinity. Thus, the good events happen in a sense more often than the bad events.

For bad events, as $r_{\gamma}$ tends to $\infty$, a disc of radius $r_{\gamma}$ should be empty whose probability is given by $e^{-\lambda \pi r_{\gamma}^2} $. Conversely, for good events, as $r_{\gamma}$ tends to 0, there should at least be a node in the disc of radius $r_{\gamma}$, an event whose probability is given by $ 1 - e^{-\lambda \pi r_{\gamma}^2} $. Thus, the difference in scale is due to the fact that the probability of occurrence of a bad event goes to zero with $r_{\gamma}$ tending to infinity, faster than the probability of occurrence of good event with $r_{\gamma}$ tending to 0.
\end{proof}

\section{Sharing Number Process}

In contrast to the SNR or the Shannon rate process, which take their values
in the continuum, the sharing number process takes a discrete set of
values and is piece-wise constant.
We shall start by evaluating the frequency of discontinuities and then introduce
an upper-bound process to be used in the sequel. We conclude this section with a study of rare events.

\subsection{Discontinuities}
In the sequel, we will use the following simplified notion of a Johnson-Mehl cell, see \cite{Miles1988}.

\begin{definition}
Consider $\Phi = \{X_i\}$ the Poisson point process of nodes on $\mathbb{R}^2 $. For any threshold $\gamma$,
the Johnson--Mehl cell $J^{(\gamma)}(X_i)$ associated with $X_i$ is the 
intersection of the Voronoi cell of $X_i$ w.r.t. $\Phi$
with a disc of radius $r_{\gamma}$ centered at $X_i$, see Fig.\ref{johnson}.
\end{definition}
Thus for a fixed threshold $\gamma$, the mobile user is covered/served at time $t$  if and only if 
it is within the Johnson--Mehl cell of its associated node $X(t)$ for the radius
$r_{\gamma}$.
We recall that:
\begin{definition}
The process $(N^{(\gamma)} (t), t \geq 0)$ is defined as the number of static users present in the Johnson-Mehl cell $J^{(\gamma)}(X(t))$ if the mobile is in $J^{(\gamma)}(X(t))$ and 0 otherwise.
\end{definition}


\begin{figure}[!t]
\centering
\includegraphics[scale = 0.45]{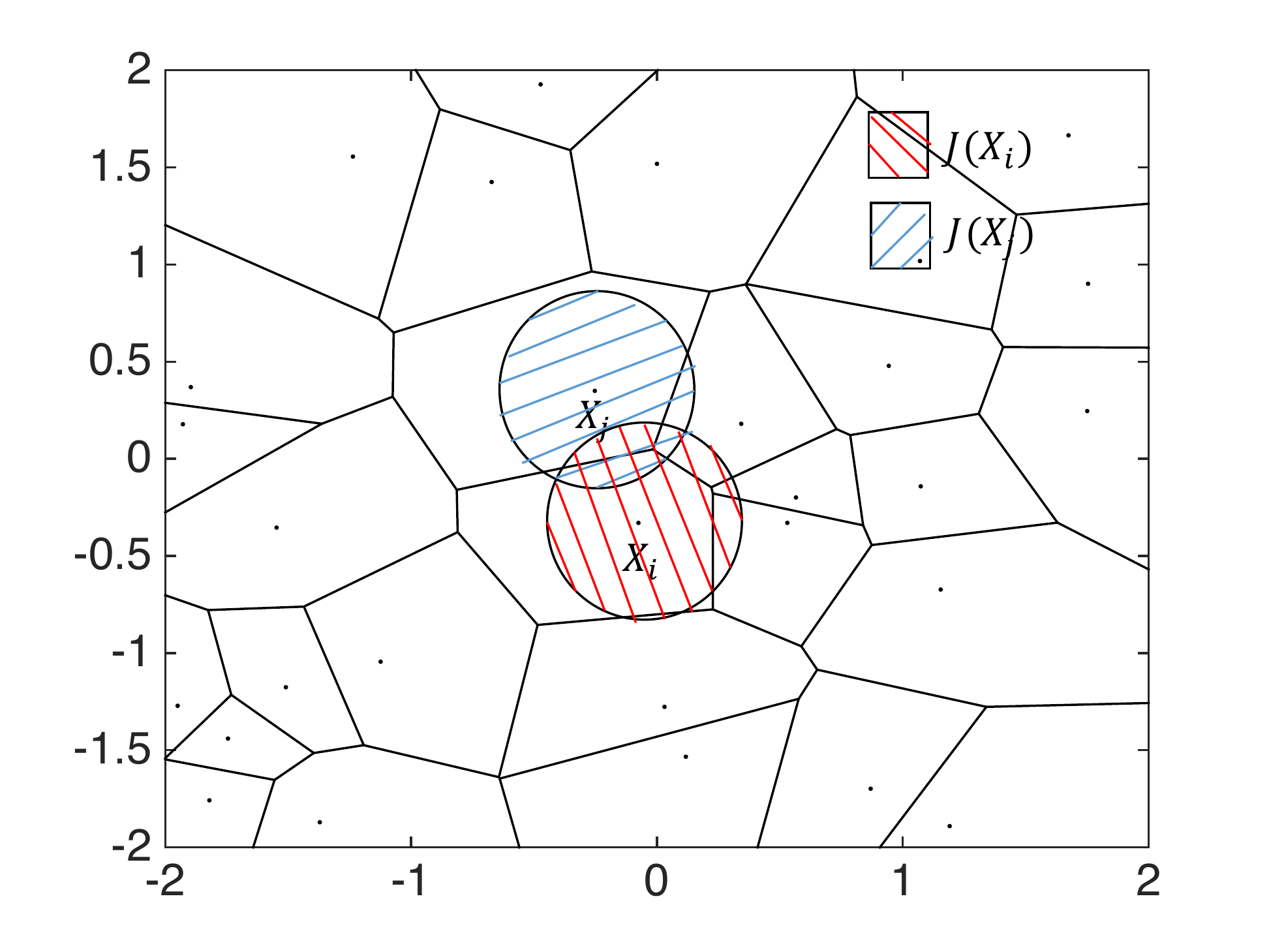}
\caption{Johnson-Mehl cells.}
\label{johnson}
\end{figure}

The process 
$(N^{(\gamma)} (t), t \geq 0)$
is an ${\mathbb N}$-valued
piece-wise constant process, with jumps 
at certain Johnson--Mehl cell edge crossings as illustrated in Figure \ref{sharingprocess}.
The value it assumes upon entering a cell is a conditionally Poisson
random variable with a parameter that depends on the area of the cell
in question. The following theorem provides an upper bound for the
jump intensity:

\begin{figure}[!t]
\centering
\includegraphics[scale= 0.4]{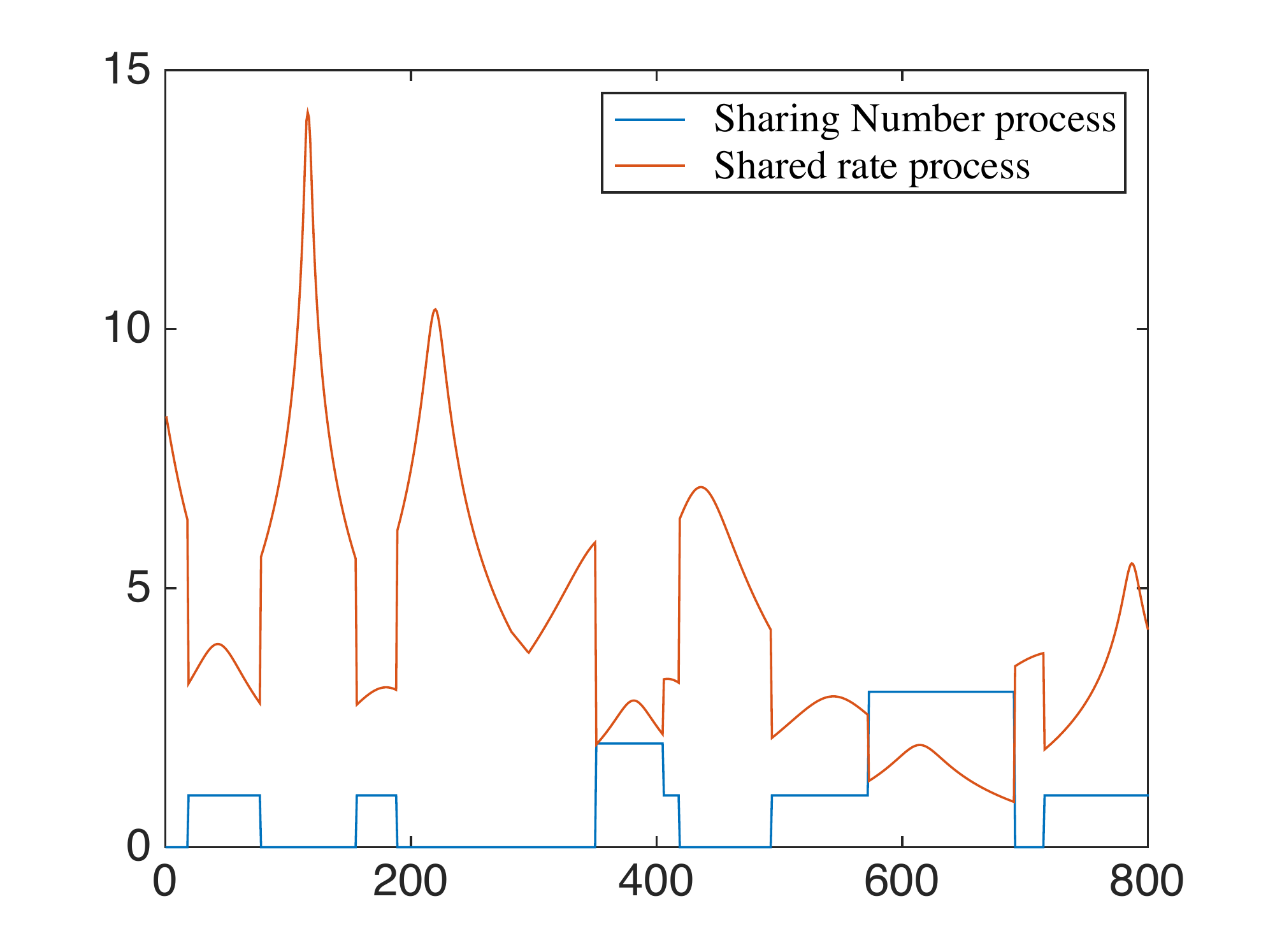}
\caption{Trace of the Sharing number and Shared rate processes}
\label{sharingprocess}
\end{figure}

\begin{theorem}
\label{thdissn}
An upper bound for the intensity of discontinuities of the sharing
number process is the intensity of the Johnson-Mehl cell edge crossings,
which is given by:
\end{theorem}
$$\frac{4v\sqrt{\lambda}}{\pi} (\mbox{erf}(\sqrt{\lambda \pi} r_{\gamma} - 2\sqrt{\lambda}r_{\gamma} e^{-\lambda \pi r_{\gamma}^2}) .$$

\begin{proof}
The intensity of cell edge crossings in a Johnson-Mehl tessellation is given in \cite{moller1992random}. Not all the cell edge crossings lead to jumps as the number of static users can be the same across two
adjacent Johnson-Mehl cells, thus the given intensity is an upper bound.
\end{proof}

Given that the value of $N^{(\gamma)}(t)$ is a non-zero constant, the {\em time of constancy} is defined as the typical amount of
time it remains at that same constant. This time of constancy
is lower bounded by the distribution of the chord length of
the motion line's intersection with a typical Johnson--Mehl cell
\cite{moller1992random}. It is a lower bound only since
the number of static users can be the same across two
adjacent Johnson-Mehl cells.

\subsection{Upper Bound for the Sharing Number Process}
Note that the two underlying processes that are used to define the shared
rate process, namely the sharing number process $(N^{(\gamma)}(t), t \geq 0)$
and the SNR process $(\mbox{SNR}(t), t \geq 0 )$, are not independent
as easily seen in the case where $\gamma=0$.
For example, assume that the mobile user experiences a low SNR, i.e.,
it is far from its associated node, Then the area of
the Voronoi cell is likely to be large. This in turn implies that the mobile is more
likely to share its node with a large number of users.

Fortunately, the sharing number process admits a simple upper
bound which is independent
of the SNR process. This upper bound is defined as follows:
\begin{definition}
Let $D(x,r_{\gamma})$ denote a disc of radius $r_{\gamma}$ centered at location $x$. Let $\hat{N}^{(\gamma)}(t) , t \geq  0, $
be the number of static users which are in the
disc $D(X(t),r_{\gamma})$ if the mobile is in $J^{(\gamma)}(X(t))$, and 0 otherwise.
\end{definition}

The stochastic process $(\hat{N}^{(\gamma)}(t) , t \geq0 )$
enjoys most of the structural
properties of $({N^{(\gamma)}}(t) , t \geq0 )$; in particular,
it is piece-wise constant and one can derive
natural bounds on the intensity of its jumps.
In addition and more importantly: 
\begin{itemize}
\item For all $t\ge 0$,
$\mathbb{P(}\hat{N}^{(\gamma)}(t) \geq N^{(\gamma)}(t)) =1 $, i.e., 
$\hat{N}^{(\gamma)}(t)$ is an upper bound for the sharing number;
this bound is tight in the regime where the area of the typical 
Voronoi cell ($\frac{1}{\lambda}) $ is large compared to the area
of the disc of radius $r_{\gamma}$;
\item When the mobile user is covered, the sharing number
$\hat{N}^{(\gamma)}(t)$
is Poisson with parameter $\xi \pi r_{\gamma}^2$; 
\item The stochastic processes $(\hat{N}^{(\gamma)}(t) , t \geq0 )$ and
$(\mbox{SNR}(t), t \geq 0 )$ are independent.
\end{itemize}

\section{Shared Rate process }

We recall that the {\em shared rate} is
\begin{equation}
S^{(\gamma)}(t) =R^{(\gamma)}(t) F^{(\gamma)}(t) = R^{(\gamma)}(t) \frac{1}{N^{(\gamma)}(t)+1},
\end{equation}
with  $R^{(\gamma)}(t)$ the Shannon rate defined in \eqref{eq:shannon} and $F^{(\gamma)}(t)$
the sharing factor at time $t$. A realization of this process is illustrated in Figure \ref{sharingprocess}.
In the sequel, we use our upper bound process on $N^{(\gamma)}(t)$ to obtain a lower bound on the shared rate: 
\begin{equation}
\label{eq:sharedrateb}
\hat{S}^{(\gamma)}(t) =R^{(\gamma)}(t) \hat{F}^{(\gamma)}(t) = R^{(\gamma)}(t)\frac{1}{\hat N^{(\gamma)}(t)+1},
\end{equation}
which is now a product of two {\em independent} random variables. 

\subsection{Shared Rate Variability}

The shared rate process is smooth for Lebesgue almost every $t$,
with a countable set of discontinuities 
and a countable set of points of discontinuities of its first derivative.
It is equal to zero when the mobile is not covered.
The point process of discontinuities
is upper bounded by the Johnson-Mehl cell edge crossings (see Theorem \ref{thdissn}).
The point process of its local maxima where the derivative is zero
is the same as that of the SNR process as given in Theorem \ref{th:thm1}.

As can be seen from \eqref{eq:sharedrateb}, the variability in the mobile user's shared rate is driven by two processes: the Shannon rate and the sharing factor process. It is of interest to understand their relative contributions. To answer this question, let us consider
the variance of $\hat S^{(\gamma)}$  which can be written using the conditional variance formula as:
\begin{equation}
\mbox{var}(\hat S^{(\gamma)}) =  \mbox{var}(R^{(\gamma)}) \mathbb{E}[\hat{F}^{(\gamma)}]^2 +\mbox{var}(\hat{F}^{(\gamma)})(\mathbb{E}[(R^{(\gamma)})^2]).
\end{equation}

Note that if the density of users increases, the variance of the
shared factor decreases whereas the variance of the Shannon rate remains constant. Thus, the variance of the shared
rate varies approximately linearly with a slope equal to
the second moment of the Shannon rate. By contrast if the density of nodes increases, the variance of both the sharing factor and Shannon rate vary making their relative contributions more complex.

Let us see empirically what are the contributions of the user and node density to the variability of shared rate. For a given node density $ \lambda = 25/\pi$ and radius of coverage $ r_{\gamma} = 500 m$, by increasing the density of static users,
the variance of the shared rate decreases linearly with 
the variance of the sharing factor.
Table \RN{2} gives numerical values evaluated by simulations.
For a given static user density $\xi = 50$ and radius of coverage of $200m$, by increasing the density of nodes the variance of both the sharing factor and Shannon rate varies. Table \RN{3} gives numerical values 
evaluated by simulations.

\begin{table}
\label{table:2}
\normalsize
\begin{center}
\begin{tabular}{ |c|c|c| } 
\hline
\textbf{$\xi$} & \textbf{var($S^{(\gamma)}$)} & \textbf{var($\hat{F}^{(\gamma)}$)} \\
\hline
5 & 0.6527 & 0.0303 \\
25 & 0.0674 & 0.0028 \\
50 & 0.012 & 0.0016 \\
75 & 0.0047 & 0.0006 \\
100 & 0.0023 & 0.0001 \\
\hline

\end{tabular}
\end{center}
\caption{Variance values when increasing user density.}
\end{table}

\begin{table}
\label{table:3}
\normalsize
\begin{center}
\begin{tabular}{ |c|c|c|c| } 
\hline
\textbf{$\lambda$} & \textbf{var($S^{(\gamma)}$)} & \textbf{var($\hat{F}^{(\gamma)}$)} & \textbf{var($R^{(\gamma)}$)}\\
\hline
5 & 0.0438 & 0.0049 & 0.3325  \\
25 & 5.1292 & 0.0576 & 03.3944  \\
50 & 0.6351 & 0.0642 & 0.9607  \\
75 & 01.2626 & 0.0723 & 0.6926 \\
100 & 0.9682 & 0.0756 & 0.2252  \\
\hline

\end{tabular}
\end{center}
\caption{Variance values when increasing node density.}
\end{table}

These results confirm that whereas the impact of increasing the other users density on mobile user's shared rate variance is clear, the result of increasing the density of nodes is more subtle.

\subsection{Rare Events}
Since the mobile user's shared rate depends on two components it is of interest to understand their relative contributions towards the events of high/low shared rate.
The proofs of the following theorems are given in Appendix.

\begin{theorem}
\label{thhsr}
The likelihood of the rare events associated with high shared rates 
is the same (up to logarithmic equivalence)
as that for the high SNR in the sense that
\begin{eqnarray}
& & \hspace{-1cm} \lim_{s\to \infty} - \frac 1 s \log\left(
\mathbb{P}(S^{(\gamma)} > s)\right)\nonumber \\
& = &\lim_{s\to \infty} - \frac 1 s \log\left(
\mathbb{P}(\log( 1+{\mathrm{SNR}})> s)\right)\nonumber =  \frac 2 \beta.
\end{eqnarray}
\end{theorem}
Notice that as a direct corollary of the last theorem,
\begin{eqnarray}
& & \hspace{-1cm} \lim_{s\to \infty} - \frac 1 s \log\left(
\mathbb{P}(S^{(\gamma)} > s)\right)\nonumber \\
& = &\lim_{s\to \infty} - \frac 1 s \log\left(
\mathbb{P}(S^{(\gamma)} >  s | N^{(\gamma)} = 0)\right)\nonumber\\
& = &\lim_{s\to \infty} - \frac 1 s \log\left(
\mathbb{P}(S^{(\gamma)} >  s , N^{(\gamma)} = 0)\right)\nonumber =  \frac 2 \beta.
\end{eqnarray}
For the rare events associated with a low shared rate, we will consider
the lower bound process $(\hat S^{(\gamma)}(t), t\ge 0)$:
\begin{theorem}
\label{thlsr}
The rare events of low shared rates are predominantly the same
as the rare events of high sharing number given that the mobile is at the
covering cell edge, in the sense that for some sequence $\{s_n\}$
such that $s_n \log (1 +Kr_{\gamma}^{-\beta}) \in \mathbb{Z} ~\forall n$
and $\lim_{n \to \infty} s_n \to \infty$: 
\begin{eqnarray}
& & \hspace{-1cm} \lim_{n\to \infty} - \frac 1 {s_n \log(s_n)}
\log\left( \mathbb{P}\left(\hat S^{(\gamma)} < \frac 1 s_n\right)\right)\nonumber \\
& = &\lim_{n \to \infty} - \frac 1 {s_n \log(s_n)} \log\left(
\mathbb{P}\left( 
\frac{\log (1 +Kr_{\gamma}^{-\beta})}
{\hat N^{(\gamma)} +1} < \frac 1 s_n 
 \right)\right)\nonumber\\
& = & \log (1 +Kr_{\gamma}^{-\beta}).
\end{eqnarray}

\end{theorem}

\begin{theorem}
Conditioned on a very good or bad shared rate, the relative contribution of the mobile user's location and the network congestion is as follows:
\begin{equation}
\lim_{s \to \infty}  \mathbb{P}(R^{(\gamma)} > s , N^{(\gamma)} = 0 | S^{(\gamma)} > s) = 1.
\end{equation}
\begin{equation}
\mathbb{P}(N^{(\gamma)} > as \log(1+\gamma) - 1 | 0 < S^{(\gamma)} <1/s) =  1.
\end{equation}

\end{theorem}

Conditioned on a very high shared rate, the mobile user has to be close to the associated node with no other static users sharing its resources. On the contrary, conditioned on the mobile being served and experiencing a very low shared rate, it has to share its associated node with number of users that is inversely proportional to the desired shared rate.

From the previous theorem, we know that high shared rates are predominantly the same as the event where the SNR is high and the sharing number is zero. Thus, for a given threshold $\delta$ for the shared rate process, we can study the asymptotic behavior of the distribution of the inter arrival time of shared rate up-crossings as $\delta \to \infty$.

\begin{corollary}
Let $\hat{Z_{\delta}}$ denote the inter arrival time for the 
good events of the lower bound process $(\hat S^{(\gamma)}(t), t\ge 0)$.
Then 
$$
\lim_{\delta \to \infty} f(r_{\delta}) e^{-\xi \pi r_{\gamma}^2} \hat{Z_{\delta}} \to \exp(1),$$
where $f(r_{\delta}) = 2 \lambda v r_{\delta}.$
\end{corollary}

\begin{proof}
The result follows from the fact that $\hat{Z_{\delta}} = \sum_{i=1}^{X} V^{r_{\delta}}$, where $X$ is geometric random variable with parameter $e^{-\xi \pi r_{\gamma}^2}$ and $ V^{(r_{\delta})}$ the typical interval time of the renewal process of up-crossings associated with the SNR process of threshold $\delta$.
\end{proof}

\section{Simulation and Validation}

In this section we evaluate when our 
mathematical model and associated asymptotic results 
are valid in more realistic settings. 
We use simulation to study the temporal variations 
of the SNR process experienced by a mobile user under 
various scenarios which are are not captured by
our analytical framework.
The model is challenged in various complementary 
ways: e.g., by adding fading and accounting for interference from other
base stations. 
In each case the objective is to determine for what parameter
values of the additional feature our simplified mathematical model 
is still approximately valid, providing robust engineering 
rules of thumb to predict what mobile users will see.

In particular, we will answer the following questions:
\begin{itemize}
\item How quickly do the SNR level up crossings converge to exponential asymptotics
as a function of the associated thresholds? 
\item Are the results obtained robust to the addition of fast fading? 
\item Are there regimes where the temporal characteristics of the SNR process are a
good proxy for the SINR process, e.g., high path loss?
\end{itemize}

We begin by introducing our simulation methodology and the
default parameters used throughout this section.
\subsection{Simulation Methodology}
We consider a user moving on a straight line (road) at a fixed velocity of 16 m/s. 
The base stations are randomly placed according to a Poisson point process 
with intensity $\lambda$ such that the mean coverage area per base station 
is that of a disc with radius 200m. Unless otherwise specified, 
we consider the path loss function given by Eq. \eqref{eq:snr}
with exponent  $\beta =4$ 
and assume that all base stations transmit with equal power of $p = 2 \mbox{W}$. 
The signal strength received by the mobile user is recomputed
every $10^{-2}$ seconds. 

We calibrate the thermal noise power to the {\em cell-edge} user. 
Let $D$ be a random variable denoting the distance from a typical 
user to the closest base station. 
Define $d_{{edge}}$ by the relation $P(D \leq d_{{edge}}) =0.9$. 
Since in our simulation setting  $P(D \leq d) = 1 - \exp(\lambda \pi d^2)$, 
we have $d_{{edge}} = \sqrt{\frac{- \ln(0.1)}{\pi \lambda}}$. 
If we fix the desired SNR at the cell edge to be $\mbox{SNR}_{{edge}}$ this
then determines the noise power to be $w = \frac{pd_{{edge}}^{-\beta}}{\mbox{SNR}_{{edge}}}$.

In the sequel we evaluate how quickly the convergence to exponential
studied in Theorem \ref{asymptotic} arises. 
To that end we compare the renormalized distributions obtained via simulation
to the reference exponential distribution with parameter 1, using
the Kolmogorov-Smirnov (K-S) test. The K-S test 
finds the greatest discrepancy between the observed 
and expected cumulative frequencies-- called the \q{D-statistic}.  
This is compared against the critical D-statistic for that sample size 
with 5\% significance level.  
If the calculated D-statistic is less than the critical one, we conclude 
that the distribution is of the expected form, see e.g. \cite{KS}. 

\subsection{Convergence of Level-crossing Asymptotics}

Theorem \ref{asymptotic} indicates that as the SNR threshold $\gamma$ in dB increases, the rescaled
distribution for up-crossings of the SNR process becomes exponential. 
The question is how large $\gamma$ needs to be for this result to hold.
To that end we simulated the level crossing process for various 
$\gamma$ and computed the D-statistic mentioned above. 
The empirical CDF for up-crossing inter-arrivals 
rescaled by $f(r_{\gamma})$ as introduced in the theorem can be seen
in e.g., Fig.~\ref{snrupnofad}. As expected we found that
as the threshold increases, the distribution 
becomes exponential, and for a threshold value of $\gamma = 50$ or
more, it is exponential with unit mean.

In practice SNR of 50 dB is not realistic for wireless users. However, as seen from Figure \ref{snrupnofad}, for moderate values of  $\gamma$ such as $0.1, 1$, the up-crossing inter-arrivals can be approximated by an exponential with parameter $1/ f(r_{\gamma})$. For $ \gamma = 1$, the empirical mean for the inter-arrival time for up-crossings is $72.4s $ and the asymptotic approximated mean i.e., $f(r_{\gamma})^{-1}$ is $63.07s$.  

\begin{figure}[!t]
\centering
\includegraphics[scale=0.45]{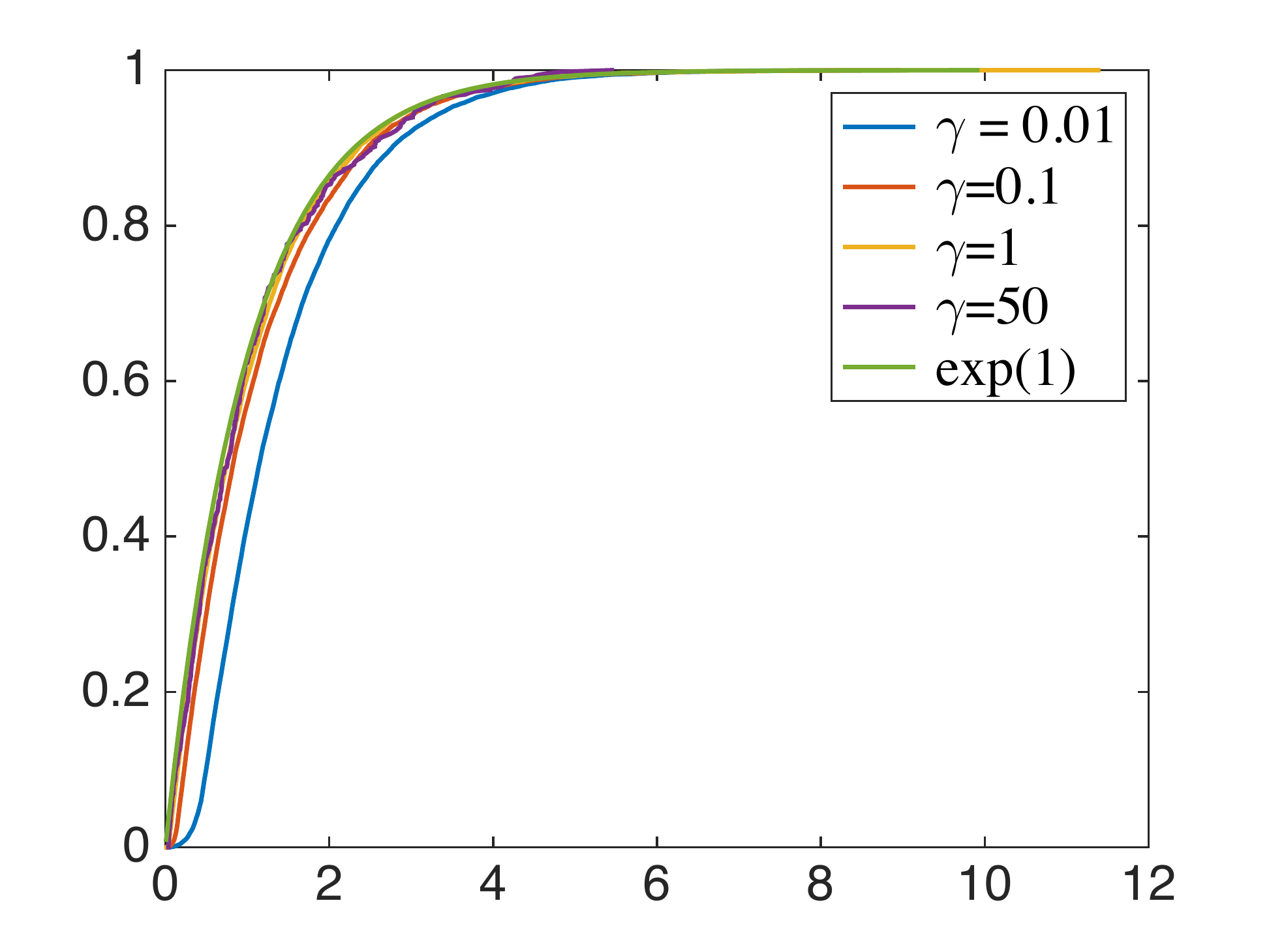}
\caption{CDF of interarrival of up-crossings for various thresholds.}
\label{snrupnofad}
\end{figure}

\subsection{Robustness of Level-crossing Asymptotics to Fading}

Next we study the effect that channel fading might have 
on the level-crossing asymptotics. 
We consider channels with Rayleigh fading with unit mean, so that the
SNR experienced by the tagged mobile user at a distance $d$ 
from the base station is given by $H p d^{-\beta} /w$, 
where $H$ is a fading random variable which is exponential 
with unit mean. The coherence time is set to $t_c = 0.423/f_d$, 
where $f_d$ is the Doppler shift given by $f_d = \frac{v}{c} f_o$ 
where $v$ is the vehicle velocity, $c$ is the speed of lights 
and $f_o = 900\mbox{MHz}$ is the operating frequency. 
This gives a coherence time $t_c =0.007 s$. Thus, fading (power) changes 
every 0.007 seconds. The SNR process with fading is illustrated in Fig.~\ref{snrfad-1}. 


Clearly when we incorporate channel fading in the SNR process, 
even when one fixes a high SNR threshold,
the level-crossing process has additional fluctuations before it goes
down again for some time, see Fig.~\ref{snrfad-1}. 
Thus to exhibit the on-off structure and asymptotics we 
consider a modified process defined as follows. 
After the first up-crossing, we suppress all subsequent up crossings (if any)
for an appropriate time period, and then look for the next up-crossing
taking place after this time. We take a 
time period for the suppression of up-crossings equal to twice the expected on time of $2E[B^{(r_{\gamma})}]$ ~\cite{Makowski}.

In order to vary the variance while keeping the mean of the fading at one, we now consider fading which is a mixture of exponentials.
For this process, we would expect that for fading with mean one,
if variance is small,  the appropriately rescaled 
inter-arrival distribution for up-crossings which
are not suppressed would once again asymptotically 
become exponential with parameter 1. In other words
we expect geometric variations associated with base station locations
to dominate channel variations. Whereas, if the fading variance is high,
one might expect the SNR threshold required to obtain convergence to 
an exponentiality to increase. Fig.~\ref{varsnr} plots such thresholds
as a function of the fading variance. 
As can be seen for fading variances exceeding $8$,
the channel variations dominate the geometric
variations and thus up-crossing asymptotics differ from Theorem \ref{asymptotic}.

 \begin{figure}[!t]
\centering
\includegraphics[width=3in]{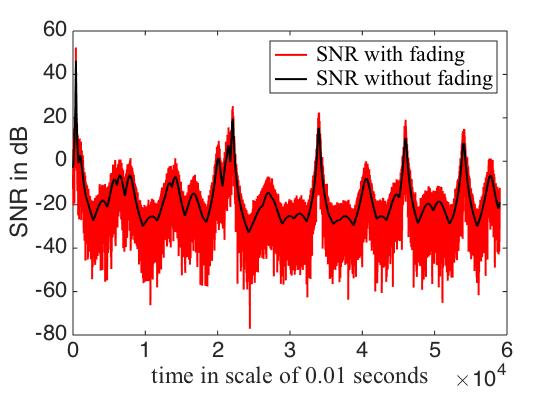}
\caption{SNR process in presence of fading with mean 1.}
\label{snrfad-1}
\end{figure}

 \begin{figure}[!t]
\centering
\includegraphics[scale=0.4]{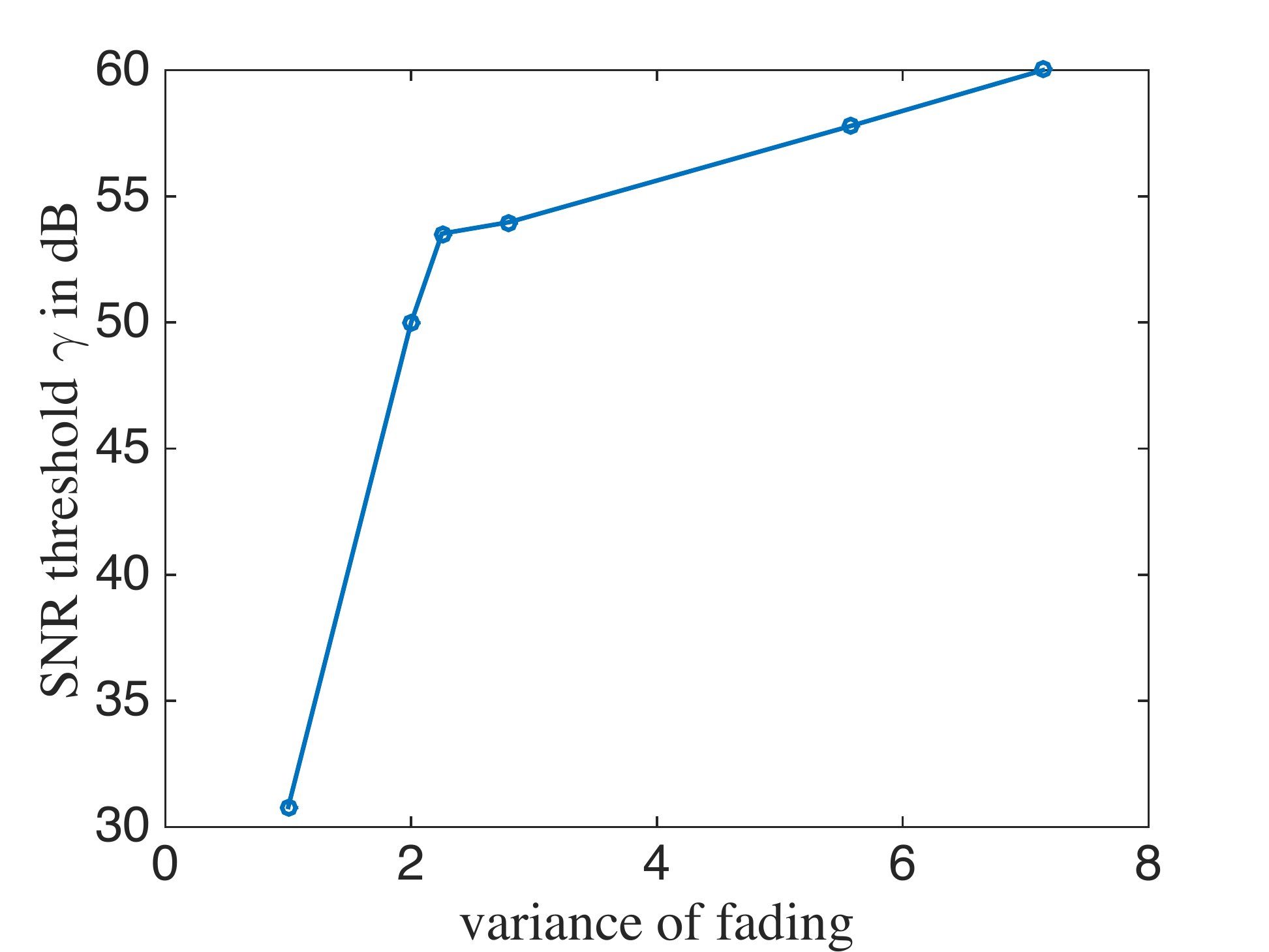}
\caption{Threshold above which the inter-arrival of up-crossing converge to exponential with parameter 1 for different variance of fading.}
\label{varsnr}
\end{figure}

\subsection{Robustness of Level-crossing Asymptotics to Interference}

So far we have focused on the SNR process. One might 
ask to what degree the Signal-to-Interference-plus-Noise Ratio (SINR)
process, shares similar characteristics. 

We first simulated the SINR process for a setting with a high path 
loss exponent  of $\beta = 4$  and found once again that 
the rescaled distribution for the up-crossing inter-arrivals 
converges to an exponential with parameter 1. The test requires a threshold $\gamma = 31.7$dB. However, as seen above, this asymptotic is already useful for moderate values of $\gamma$.
We then evaluated, for different path loss $\beta$, what
threshold values were needed to obtain a similar convergence.
As shown by Fig.~\ref{betasnr}, the threshold in question
increases as $\beta$ decreases. Further we found
that for $\beta < 3.5$, we no longer have the desired convergence
property. 
In summary, for high path-loss exponents $\beta =3.5-4$, the up-crossing 
asymptotics for the SNR and SINR processes are similar.

\begin{figure}[!t]
\centering
\includegraphics[scale=0.4]{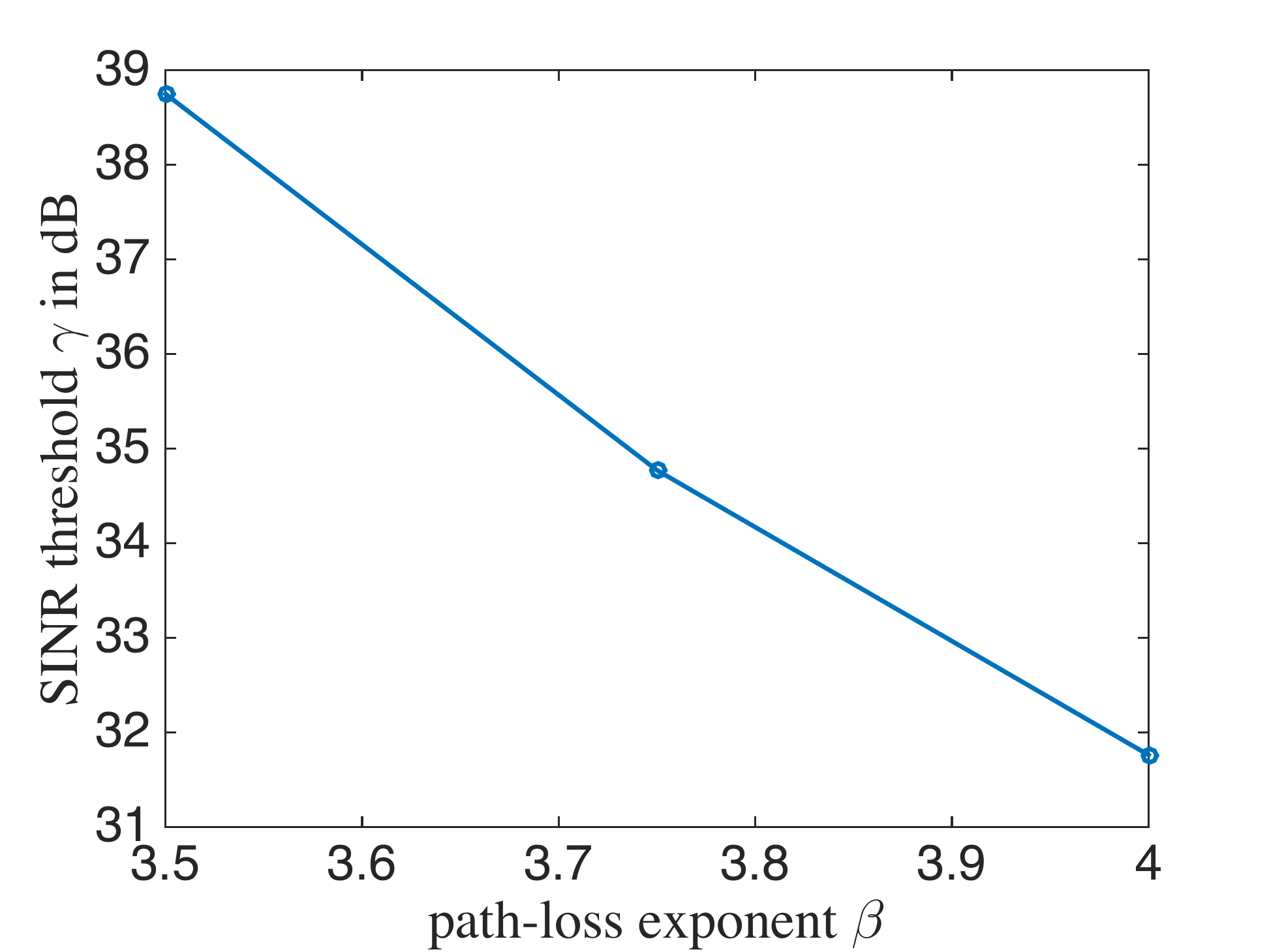}

\caption{Threshold above which the inter-arrival of up-crossing converge to exponential with parameter 1 for different path loss exponent.}
\label{betasnr}
\end{figure}

\section{Applications}
In this section we consider our models to evaluate application-level performance of mobiles in a shared wireless network.
In particular we consider two very different wireless scenarios: (1) video streaming to mobiles sharing a cellular or Wifi network and (2) large file downloads using Wifi.

\subsection{Stored Video Streaming to Mobile Users}

Let us consider a scenario where mobile users are viewing stored videos which are streamed over a sequence of wireless downlinks. The users are distributed according to a Poisson point process of density $\xi$ and are moving independently of each other. Consider a policy where a mobile user is served by a node only if the SNR experienced is greater than a threshold $\gamma$. Thus, the nodes serve the mobile users within the radius of coverage $r_{\gamma} =(p/(\gamma w))^{1/\beta}$.  A lower threshold $\gamma$ corresponds to a lower transmission rate (when served), a higher probability of coverage and sharing with a large number of other mobile users. Conversely a higher threshold implies higher transmission rate and sharing with fewer other mobile users. For simplicity we consider rebuffering as the primary metric for user's video quality of experience \cite{ZhD14}.

The playback buffer state of the tagged mobile user moving at a fixed velocity $v$ on a straight line, can be modeled as a fluid queue. The arrival rate to which alternates between
an average ergodic rate, $h(\gamma) = \mathbb{E}[R^{(\gamma)} | \mbox{SNR} > \gamma]$ and zero depending on whether the mobile is being served or not. Let $\eta$ denote the video playback rate in bits. Hence, as long as the buffer is non-empty,
the fluid depletion rate of the queue is $\eta$. Re-buffering of the video is directly linked to the proportion of time the playback
buffer is empty, which is given by $1-\rho(\gamma,\xi)$, where the {\em {load factor}}, $\rho(\gamma,\xi)$, \cite{BB3} of the queue is:

\begin{equation}
\label{eqrx}
\rho(\gamma,\xi) = \frac{ \nu^{(\gamma)}  h(\gamma)}{\eta}  \mathbb{E}\bigg[\frac{1}{N_p^{(\gamma)}+1}\bigg].
\end{equation}
where
$\nu^{(\gamma)}$ is the probability that the alternating renewal process associated with the arrival rate to the fluid queue is \q{on}.
 

The first natural question one can ask is whether there is a choice of $\gamma$
such that the fluid queue is unstable, thus ensuring no rebuffering in
the long term. In other words, does there exist a $\gamma>0$ such that
$\rho(\gamma,\xi)>1$? Given our policy and the metric for quality of experience, the network provider has liberty to choose a lower value of threshold,$\gamma$, as long as the typical mobile user in the long run experiences no rebuffering.

Now, for simplicity let us consider a constant rate $\kappa = a \log\left(1 + \gamma \right) \mathbb{E}\bigg[\frac{1}{N_p^{(\gamma)}+1}\bigg]$ instead of the average ergodic rate. This constant bit rate is when the network does not rely on adaptive coding/decoding. For additional motivation for this scenario, see \cite{ZhD14}. The load factor $\rho(\gamma ,\xi)$ is given by


 \begin{equation}
 \begin{split}
\label{eqrx2}
\rho(\gamma,\xi) =
\frac {a
\log(1+\gamma) \left(1-e^{\frac{-b}{\gamma^{2/\beta}}}
\right)}
{\eta} \mathbb{E}\bigg[\frac{1}{N_p^{(\gamma)}+1}\bigg],
\end{split}
\end{equation}
where
$b= \lambda \pi \left( \frac{p} w\right)^{\frac 2\beta}$ and $\mathbb{E}[1/(N_p^{(\gamma)}+1)]$ can be calculated by numerical integration as described in previous section.

It is easy to check that the function
$\rho(\gamma,\xi)$
has a unique maximum $\gamma^*$ on $(0,\infty)$. A plot of (\ref{eqrx2}) and the value of $\gamma^*$
are illustrated in Fig. \ref{Myfig}. Notice that for these parameters,
the value of $\gamma^*$ increases with $\xi$.

\begin{figure}[h]
\centering
\includegraphics[scale=0.4]{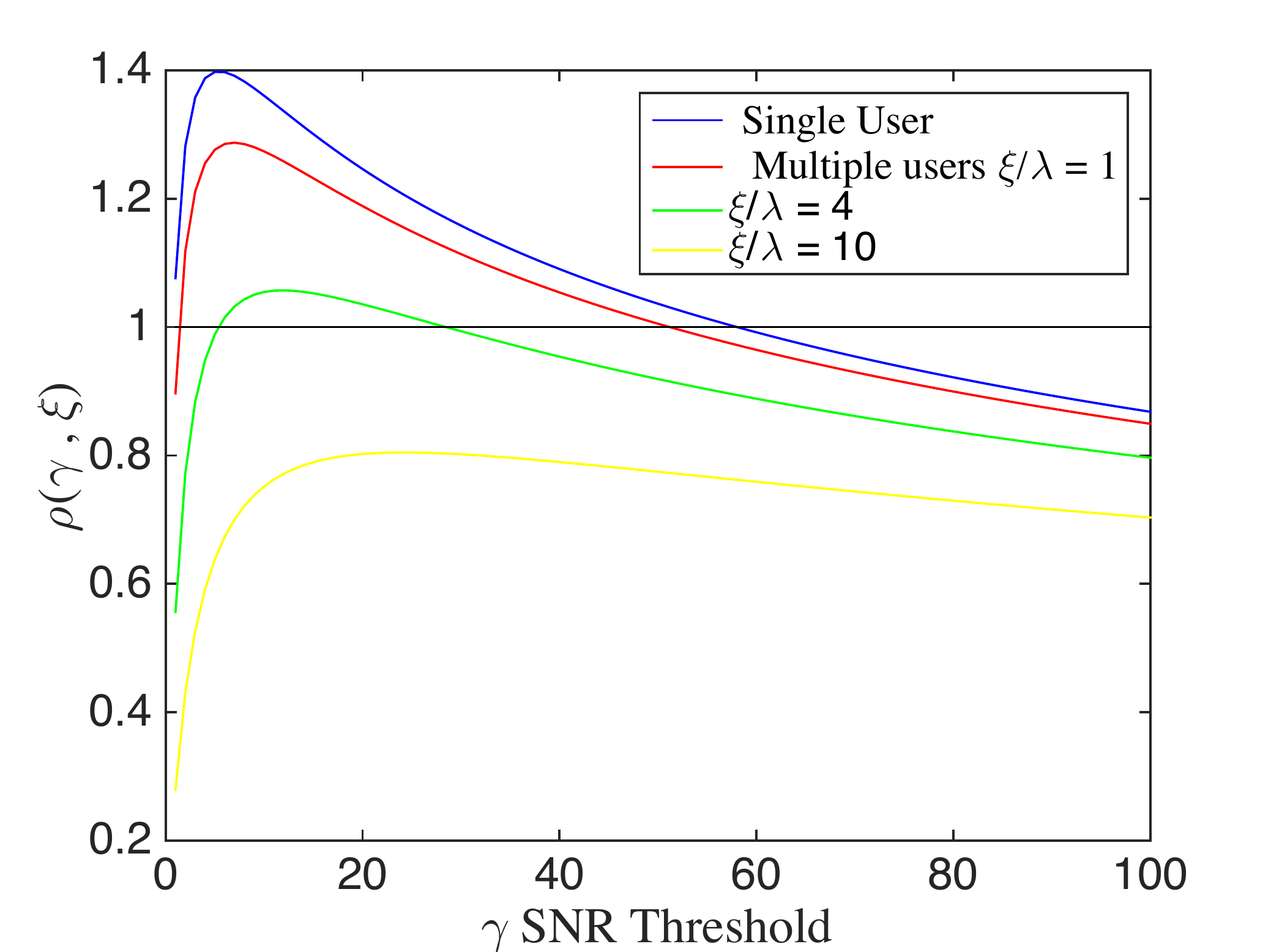}
\caption{Load factor of the fluid queue for a single user (top curve) 
and for a positive density of users as a function of $\gamma$ (curves below the
top curve).
For the latter curves, the number of users per base station are $\xi/\lambda=1,4$ and 10 from top to bottom.
Here $b=1$. All functions can be multiplied by an arbitrary
positive constant when playing with $a$ and $\eta$.
}
\label{Myfig}
\end{figure}

For a given base station density $\lambda$ and density of users $\xi$, one can evaluate the SNR threshold value $\gamma^*$ for which the load factor $\rho$ is maximum. Fig.\ref{gamma_xi} illustrates the level set curve of $\rho(\gamma^* , \xi) =1$ for various values of $\lambda$ and $\xi$. In this setup, given the video consumption rate $\eta$, it is possible to answer questions like what is the minimum density of base stations required to serve a certain density of users such that the video streaming is uninterrupted for all the users.

\begin{figure}[h]
\centering
\includegraphics[scale =0.45]{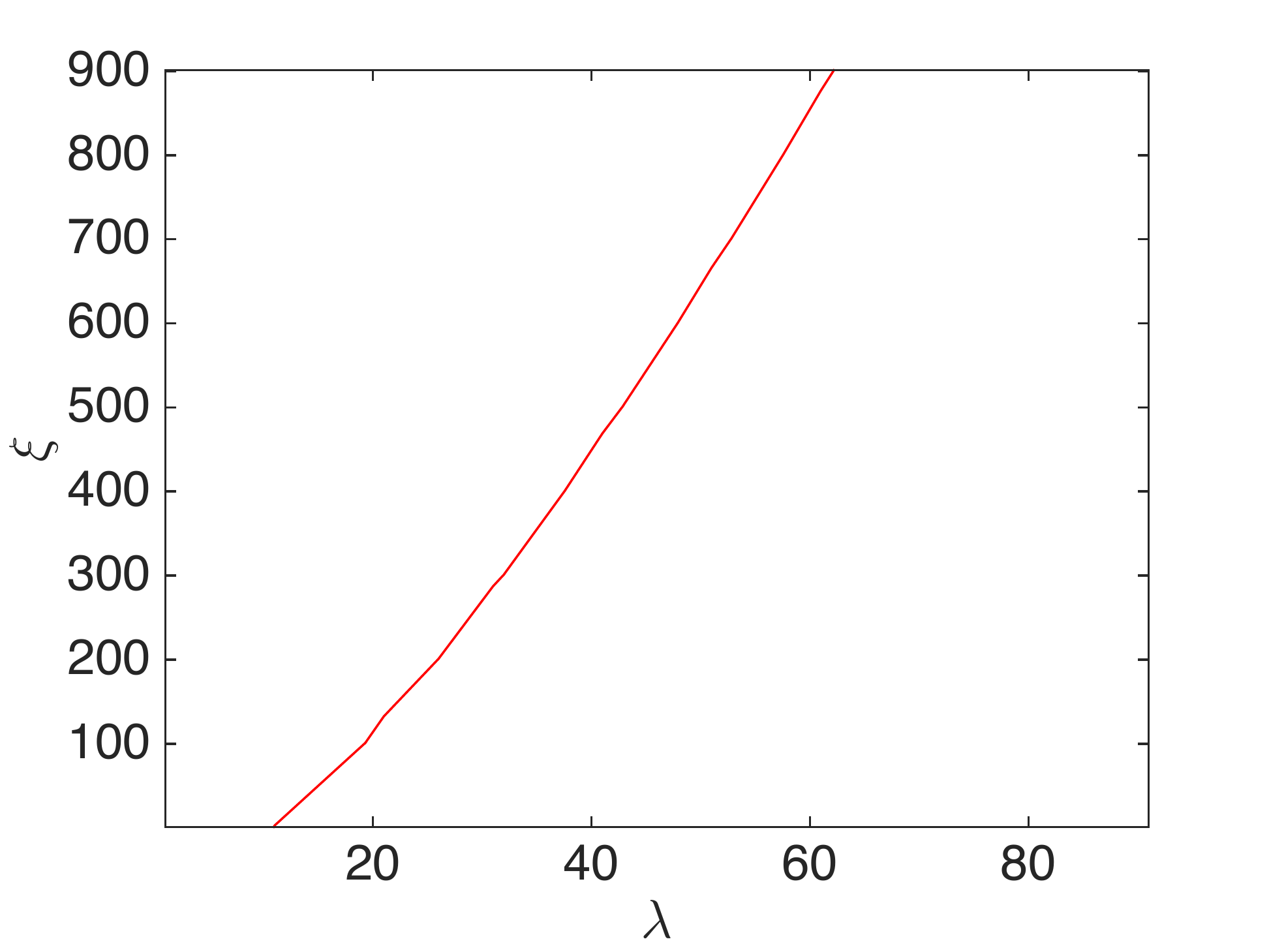}
\caption{Level set curve of $\rho(\gamma^*,\xi) =1$ for an arbitrary
positive constants $a$ and $\eta$.
}
\label{gamma_xi}
\end{figure}

\begin{remark}
In the case where there exists no threshold $\gamma$ for
which the condition for no long term rebuffering
is satisfied and $\kappa > \eta$, the fluid queue
alternates between busy and idle
period representing the periods when the
video is un-interrupted or frozen, respectively.
The distribution of the on periods $B^{(\gamma)}$ and
that of the off periods $I^{(\gamma)}$ of 
of the $M/GI/\infty$ queue discussed in Section
determine the distribution of the busy period $B_f$
and that of the idle period $I_f$ of the fluid queue.
When denoting by $\kappa$ the constant input rate
during on period and by $\eta$ the  constant output rate
when the queue is non-empty, 
the Laplace transform of $B_f$ is given by \cite{gupta2008fluid}
\begin{equation}
\mathcal{L}_{B_f}(s) = \mathcal{L}_{B^{(\gamma)}}(s \sigma + \lambda(\sigma -1)(1-\mathcal{L}_{B_f}(s))),
\end{equation}
where $\sigma = \kappa/ \eta > 1$, and the idle period
$I_f$ has an $\mbox{exp}(\lambda^{(\gamma)})$ distribution.
\end{remark}

\subsection{Wifi Offloading}
WiFi offloading helps to improve spectrum efficiency and reduce cellular network congestion. One version of this scheme is to have mobile users opportunistically obtain data through WiFi rather than through the cellular network. Offloading traffic through WiFi has been shown to be an effective way to reduce the traffic on the cellular network. WiFi is faster and uses less energy to transmit data when there is a connection. 

Let us consider a scenario where the mobile users download a large file from the service provider,
relying on Wifi hotspots, distributed according to a Poisson point process of intensity
$\lambda$, rather than on cellular base stations. The users are distributed according to a Poisson point process of density $\xi$ and are moving independently of each other. Assume that the Wifi hotspots have a fixed coverage area i.e., the mobile user connects to Wifi only if it
is within a certain distance $r$ from the hotspot. Thus, higher the density of hotspots, $\lambda$, deployed by the provider the better the performance experienced by mobile users that rely only on them. We consider the time it takes to complete the
download as the primary metric for user's quality of experience.

Consider again the case without adaptive coding/decoding and the tagged mobile user moving on a straight line at constant velocity $v$.Then, 
the shared rate experienced by the mobile user is the constant $\kappa = a \log\left(1 + \gamma \right) \mathbb{E}\bigg[\frac{1}{N_p^{(\gamma)}+1}\bigg]$ as
defined above. In addition, the mobile user experiences an alternating on and off 
process, as characterized in Theorem \ref{level crossing characterization}.

Below, for the sake of mathematical simplicity, we assume
that the file size $F$ is exponential with parameter $\delta$
and that the mobile user starts to download the file at
the beginning of an on period.
Let $T$ be a random variable denoting the time taken to download the file.
Consider the event $J = \{F >\kappa B^{(\lambda)}\}$ and let
\begin{equation}
\alpha =\mathbb{P}(J) = \mathbb{P}(F > \kappa B^{(\lambda)}) = {\cal L}_{B^{(\lambda)}}(\delta \kappa).
\end{equation}
Here and below, ${\cal L}_{X}(s)$ denotes the Laplace transform
of the non-negative random variable $X$ at point $s$,
and $B^{(\lambda)}$ is the random variable representing the length of a
typical on interval seen by the mobile (see Theorem \ref{level crossing characterization}). 

Now, define the non-negative random variables $X$ and $Y$ by their c.d.f.s
\begin{eqnarray}
\mathbb{P} (X < x) & = & \frac{1}{\alpha}  \int_0^{x} e^{-\delta \kappa z} f_{B^{(\lambda)}}(z) dz
\\
\mathbb{P} (Y < y) & = & \frac{1}{1- \alpha}  \int_0^{y} (1-e^{-\delta \kappa z}) f_{B^{(\lambda)}}(z) dz.
\end{eqnarray}
Notice that 
\begin{eqnarray}
{\cal L}_{X}(s) & =  & \frac 1 \alpha {\cal L}_{B^{(\lambda)}} (s+ \delta \kappa )
\\
{\cal L}_{Y}(s) & =  &  \frac{1}{1- \alpha}
\left(1-{\cal L}_{B^{(\lambda)}} (s+ \delta \kappa )\right).
\end{eqnarray}
The following representation of the Laplace transform of $T$
is an immediate corollary of the on-off structure:
\begin{theorem}
Consider a network of Wifi hotspots distributed according to a Poisson point process of intensity $\lambda$ and radius of coverage $r$ shared by mobile users distributed independently by a Poisson point process of intensity $\xi$. Assuming that the tagged mobile user starts to download the file at
the beginning of an on period, the Laplace transform of the time taken to download a file of size $F \sim \exp(\delta)$ is given by
\begin{equation}
\label{equationl}
{\cal L}_{T}(s)
= \frac{(1-\alpha){\cal L}_{Y}(s)}{1 -\alpha{\cal L}_{X}(s) \frac {2\lambda v r}{2\lambda v r+ s}}.
\end{equation}

\end{theorem}
It is remarkable that the Laplace transform of $T$ admits a quite simple expression in terms of that of $B^{(\lambda)}$. Other and more general file distributions can be handled as well when using classical tools of Laplace transform theory. Note that this setting also leads to interesting optimization questions such as the optimal density of Wifi hotspots needed to be deployed for lower expected download times.

\section{Variants}
In this section, we consider two  generalizations of the previous
framework: (1) we move from sharing with static users to sharing with mobile users and (2) we move from homogeneous to heterogeneous infrastructures

\subsection{Sharing the Network with Mobile Users}

Until here we have considered a tagged mobile user sharing the network with static users. We now consider two scenarios: (1) that where the other users are mobile and are initially distributed according to a homogeneous Poisson point process and (2) that where they are mobile but restricted to a random road network, i.e., form a Cox process.

\subsubsection{Homogeneous Poisson case}

Consider the case where other users sharing the network are initially located according to a homogeneous Poisson point process of intensity $\xi$, and subsequently exhibit arbitrary independent motion. This is a possible model for pedestrian motion. It follows from the displacement theorem for Poisson point processes \cite{BB} that the other users at any time instant will remain a Poisson point process of intensity $\xi$.

As considered before, the users are served only if they are at a distance less than $r_{\gamma}$ from the closest node. Consider a tagged mobile user moving at a fixed velocity along a straight line. Let us define the mobile sharing number process $(N_p^{(\gamma)}(t) , t \geq 0)  $ as the number of users sharing the node associated with the tagged user when it is served and zero otherwise. Thus, at any given time $t$, the tagged user shares its resources with a random number of users $N_p^{(\gamma)}(t)$  which is Poisson with a parameter depending on the area of the Johnson-Mehl cell.

For simplicity, we define the shared rate process $(S_c^{(\gamma)}(t) , t \geq 0)$ as
 \begin{equation}
\label{eqrx}
S_c^{(\gamma)}(t) = 
\begin{cases}
a \log\left(1 + \gamma \right) \mathbb{E}\bigg[\frac{1}{N_p^{(\gamma)}+1}\bigg] & \text{if}
~~ L(t) \leq r_{\gamma}, \\
0~~ &\text{otherwise}.
\end{cases}
\end{equation}

Since the distribution of the area of the Voronoi cell is unknown, we approximate $N_p^{(\gamma)}$ to be Poisson with parameter $\xi \mathbb{E}[\hat{J}^{(\gamma)}]$, where $\hat{J}^{(\gamma)}$ denotes the area of the Johnson-Mehl cell of radius $r_{\gamma}$, conditioned that the tagged user is within a distance $r_{\gamma}$ from the associated node, which introduces an additional bias. We can compute the expectation with the help of integral geometry as shown in Appendix F. 

We found the value of the expectation using numerical integration and compared the mean number of users $\mathbb{E}[N_p^{(\gamma)}] = \xi \mathbb{E}[\hat{J}^{(\gamma)}] $ with the sample mean obtained from simulation for various parameters $\lambda ,r_{\gamma}$. To validate the approximation of $N_p^{(\gamma)}$ by a Poisson random variable with parameter $\xi \mathbb{E}[\hat{J}^{(\gamma)}]$, we compared the value of $\mathbb{E}[1/(N_p^{(\gamma)}+1)]$ calculated using numerical integration with that of simulations. We found that  the calculated value of the expectation is within the 95\% confidence interval of the simulated mean.

\subsubsection{Cox process}

Let us consider a population model where roads are distributed according to a Poisson line process of intensity $\lambda_r$  on $\mathbb{R}^2$ \cite{kingman}. Then independently on each road, we consider users distributed according to a stationary Poisson point process of intensity $\lambda_t$. This is known as Cox process and we denote it by $\Phi_u$ \cite{kingman}. This model can be used to represent a car motion in a road network.

For a line $L$ of the line process, let us denote the orthogonal projection of the origin $O$ on $L$ by $(\theta,r)$ in polar coordinates. For $\theta \in [0,\pi)$ and $r \in \mathbb{R}$,  $(\theta,r)$ is unique. Thus, a Poisson line process with intensity $\lambda_r$ is the image of a Poisson point process with the same intensity on half-cylinder $[0,\pi) \times \mathbb{R}$.

Suppose all users on the roads are moving arbitrarily but independently from each other. Thus, at any instant the distribution of users on a given road remains Poisson. Consider a tagged user moving along a given road, then we have the following theorem by \cite{morlot2012population}.

\begin{theorem}
$\Phi_u$ is stationary, isotropic, with intensity $ \pi \lambda_r \lambda_t$. From the point of view of the tagged user i.e., under the Palm distribution the point process is the union of three counting measures: (1) the atom at $O$, (2) an independent $\lambda_t$-Poisson point process on a line through $O$ with a uniform independent angle and (3) the stationary counting measure $\Phi_u$.
\end{theorem}

Following our previous framework, let us define an another sharing number process $(N_d^{(\gamma)}(t) , t \geq 0)  $ as the number of users sharing the node associated with the tagged user when it is served and zero otherwise. 

\textit{Evaluation of the mean of $N_d^{(\gamma)}$}. Suppose the tagged user is at the origin $O$. Let the associated node $X(t)$ be at a distance $x$ from the origin. Let $d(0,\theta)$ be a line through the origin with $\theta$ uniform from $[0,\pi)$. From the aforementioned theorem, the number of sharing users $N_d^{(\gamma)}$ can be split into two terms: $N_s$ denoting the number of sharing users from stationary $\Phi_u$ and $N_l$ denoting the number of sharing users on the line $d(0,\theta)$.

Given any convex body $Z$,  $\mathbb{E}[N_s(Z)]$ is given by $\pi \lambda_r \lambda_t \mbox{area}(Z)$ \cite{baccelli1997stochastic}. Let $ \hat{J}^{(\gamma)}$ denote the area as defined before. Thus,
$$
\mathbb{E}[N_s] = \pi \lambda_r \lambda_t \mathbb{E}[\hat{J}^{(\gamma)}],
$$
where, $\mathbb{E}[\hat{J}^{(\gamma)}]$ is evaluated using integral geometry (see Appendix F).

Now, let $l(0,\theta)$ denote the length of the line $d(0,\theta)$ in the area $\hat{J}^{(\gamma)}$. Then, $N_l$, the number of sharing users on the line $d(0,\theta)$, is Poisson with parameter $\lambda_t l(0,\theta)$. Thus, $\mathbb{E}[N_l] = \lambda_t \mathbb{E}[l(0,\theta)] $.

One can evaluate $ \mathbb{E}[l(0,\theta)] $ using integral geometry (see Appendix G).
We have,
 \begin{equation}
\mathbb{E}[N_d^{(\gamma)}] = \mathbb{E}[N_s]+ \mathbb{E}[N_l] = \pi \lambda_r \lambda_t \mathbb{E}[\hat{J}^{(\gamma)}] + \lambda_t \mathbb{E}[l(0,\theta)].
 \end{equation}
Thus, the mean number of users sharing the tagged user's association node is larger when the users are distributed according to a Cox process than when the users are distributed according to a Poisson point process, assuming that both have the same mean spatial intensity.

\subsection{Mixture of Pedestrian and Road Network}

Suppose now we have two types of users : drivers who stay on roads, and pedestrians which are unconstrained. If pedestrians are supposed to follow a Poisson point process, from their point of view, the number of sharing users corresponds to the sum of a stationary Poisson point process and a stationary Poisson line process. On the other hand, from a driver's point of view, the number of sharing users corresponds to the same sum, but in addition with a Poisson point process on a road passing through the driver.  Thus, the mean number of sharing users is always greater for drivers. Thus, pedestrians are likely to share its node with fewer users than drivers.

\subsection{Heterogeneous Networks}

Let us consider a deployment of micro-base stations $\hat{\Phi} = \{\hat{X}_1,\hat{X}_2, ..\} $  distributed according to some homogeneous Poisson process of intensity $\hat{\lambda}$ independent of the existing macro-base stations $\Phi =  \{X_1,X_2, ..\} $. Assume that all micro-BS transmit at a fixed power $\hat{p}$. 

Let us consider a mobile user moving at a fixed velocity $v$ along a straight line. For a given SNR threshold $\gamma$, the mobile is served by a micro-base station if its distance from its closest micro-BS is less than $\hat{r}_{\gamma} = (\hat{p}/w\gamma)^{1/\beta}$. Otherwise it is served by a macro-base station provided its distance from the closest macro-BS is less than $r_{\gamma}$.

Note that the SNR level crossing process as previously defined is again an alternating renewal process which now depends on the heterogeneous resource deployment.

\begin{theorem}
\label{hetnet}
For heterogeneous networks with preferential association to micro base stations, the probability that the stationary SNR level crossing process seen by a tagged user is \q{on} is $1-e^{-\pi (\lambda r_{\gamma}^2+ \hat{\lambda} \hat{r}_{\gamma}^2)}$. Also, the mean time for which the process is \q{on} is given by 
\begin{equation}
 \frac{e^{\pi (\lambda r_{\gamma}^2+\hat{\lambda} \hat{r}_{\gamma}^2)}-1}{2v (\lambda r_{\gamma} + \hat{\lambda}\hat{r}_{\gamma})}.
\end{equation}
\end{theorem}

\begin{proof}

In order to characterize the SNR level crossing process, we establish a connection to a Boolean model. Assume that the mobile user is moving with unit velocity. Let $B(X_i, r_{\gamma}) $ denote the closed ball of radius $r_{\gamma}$ centered at $X_i$ and $B(\hat{X}_i, \hat{r}_{\gamma}) $ a closed ball of radius $ \hat{r}_{\gamma}$ centered at $\hat{X}_i$. The union of all these closed balls forms a Boolean model  
\begin{equation}
{\cal E} = \bigg(\cup_{X_i \in \Phi} B(X_i, r_{\gamma})\bigg) \cup  \bigg( \cup_{\hat{X}_i \in \hat{\Phi}} B(\hat{X}_i, \hat{r}_{\gamma})\bigg).
\end{equation}

Now, assume that $\cal E$ is intersected by the directed line $\vec{l}$. Note that the Boolean model under consideration has independent convex grains and thus the intersection $\cal E$ $\cap$ $\vec{l}$ yields an alternating sequence of \q{on} and \q{off} periods which are independent. Let $B_h^{(\gamma)}$ and $I_h^{(\gamma)}$ be random variables denoting the length of a typical on on and off periods respectively. 

The distribution of the length of off period $I_h^{(\gamma)}$  is easy to establish using the contact distribution functions and is exponential with parameter $\lambda^*$ \cite{Miles1988}:
$$
f_{I_h^{(\gamma)}}(l) = \lambda^* e^{\big(- \lambda^* l\big)} , 
$$
where, $\lambda^*$ is $2 (\lambda + \hat{\lambda}) \mathbb{E}[R_h^{(\gamma)}]$. Here $R_h^{(\gamma)}$ is the random variable denoting the radius of the closed ball and is given by

\begin{equation}
\label{randomradius}
R_h^{(\gamma)} = 
\begin{cases}
& r_{\gamma} ~\mbox{w.p.} ~~\frac{\lambda}{\lambda + \hat{\lambda}} \\
& \hat{r}_{\gamma} ~ \mbox{w.p.} ~~\frac{\hat{\lambda}}{\lambda +\hat{\lambda}}. \\
\end{cases}
\end{equation}

Thus, the mean off period under the assumption of unit velocity is
$$
\mathbb{E}[I_h^{(\gamma)}] = \frac{1}{\lambda^*} = \frac{1}{2(\lambda r_{\gamma} + \hat{\lambda} \hat{r}_{\gamma})}.
$$

Now, in the stationary regime, the probability that the mobile user is either within a distance $\hat{r}_{\gamma}$ from a micro-BS or a distance $r_{\gamma}$ from a macro-BS i.e., it's SNR level crossing process is \q{on}, is given by the volume fraction\cite{Miles1988}:
\begin{equation}
\label{volumefraction}
c = 1 - e^{(-(\lambda + \hat{\lambda}) \bar{V})},
\end{equation}
where
\begin{equation}
\bar{V} = \mathbb{E}[v_d(\mbox{typical grain})]  = \pi \mathbb{E}[(R_h^{(\gamma)})^2] = \frac{ \pi (\lambda r_{\gamma}^2 + \hat{\lambda}\hat{r}_{\gamma}^2)}{\lambda +\hat{\lambda}}.
\end{equation}
Thus, $c = 1-e^{-\pi (\lambda r_{\gamma}^2+\hat{\lambda} \hat{r}_{\gamma}^2)}.$ The probability evaluated above does not depend on the velocity of the mobile user and thus holds for any constant velocity $v$.

The mean on period $\mathbb{E}[B_h^{(\gamma)}]$ can be evaluated using the following relation
$$
c = \frac{\mathbb{E}[B_h^{(\gamma)}]}{\mathbb{E}[B_h^{(\gamma)}] + \mathbb{E}[I_h^{(\gamma)}]},
$$

which results in
$$
\mathbb{E}[B_h^{(\gamma)}] =  \frac{e^{\pi (\lambda r_{\gamma}^2+\hat{\lambda} \hat{r}_{\gamma}^2)}-1}{2 (\lambda r_{\gamma} + \hat{\lambda}\hat{r}_{\gamma})}.
$$

Since, the mobile user is moving with a fixed velocity $v$, the mean time for which the mobile user is \q{on} is given by $\frac{\mathbb{E}[B_h^{(\gamma)}]}{v}.$

\end{proof}

These results provide an analytical characterization of the impact of heterogeneous densification on the mobile user's temporal performance. Let us now compare the performance improvement seen by the mobile user in a heterogeneous network as compared to that of a homogeneous network. The graphs in Fig. \ref{fhmeanontime} and Fig. \ref{fhvolumefraction} illustrate the difference in the expected on-times and volume fraction of the networks respectively.

\begin{figure}[h]
\centering
\includegraphics[width=3.5in]{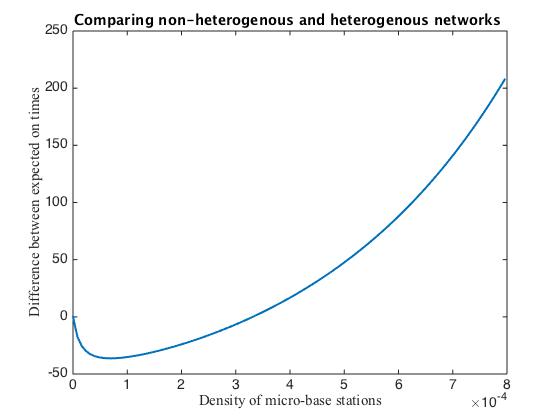}
\caption{Comparing mean-on time for heterogeneous and homogeneous networks.}
\label{fhmeanontime}
\end{figure}

\begin{figure}[h]
\centering
\includegraphics[width=3.5in]{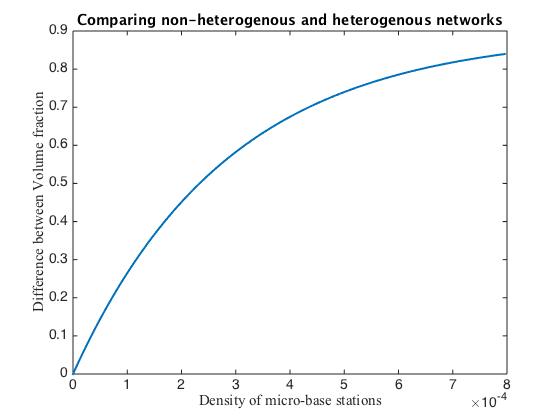}
\caption{Comparing volume fraction for heterogeneous and homogeneous networks.}
\label{fhvolumefraction}

\end{figure}

Notice that as micro-base stations are added, the expected on time decreases and later increases. The initial decrease is due to the inclusion of many relatively small length on-times resulting from the micro-base stations in the voids of the homogeneous network. However, the volume fraction increases monotonically with the addition of micro-base stations.

\subsection{Stored Video Streaming in Heterogeneous Networks }

Consider a scenario as discussed before where mobile users are viewing a video 
being streamed over a sequence of wireless downlinks, but now in heterogeneous network with micro and macro BS.
In this setting once again we consider a fluid queue representing the tagged mobile user's
playback buffer state similar to previous section and ask whether there is a choice of $\gamma$
such that the fluid queue is unstable, i.e., ensures no rebuffering in
the long term. Let

\begin{equation}
\label{eqrx1}
\rho(\gamma,\xi) =
\frac {a
\log(1+\gamma) \mathbb{P}(\mbox{on})} 
{\eta}\left(\mathbb{E}\bigg[\frac{1}{N_p^{(\gamma)}+1}| \mbox{on}\bigg]\right).
\end{equation}

Assuming our approximation is valid, we need to find $E[\frac{1}{N_p^{(\gamma)}+1}]$ in the case of heterogeneous network. Let events $G$ and $H$ be that the tagged user is served by micro BS and macro BS respectively.

\begin{equation}
\mathbb{E}\big[\frac{1}{N_p^{(\gamma)}+1}| \mbox{on}\big] = \frac{ \bigg(\mathbb{E}\big[\frac{1}{N_p^{(\gamma)}+1} | G\big]  \mathbb{P}(G)+\mathbb{E}\big[\frac{1}{N_p^{(\gamma)}+1}|H\big] \mathbb{P}(H)\bigg) }{ \mathbb{P}(\mbox{on})}.
\end{equation}

Since, the micro and macro base stations are distributed independently, the mobile user experiences two independent alternating renewal processes. Thus, in stationary regime, the probability that mobile is served by macro BS is the product of the probabilities that mobile is \q{on} period of alternating renewal process of macro BS and in \q{off} period of that of micro BS.

$$ \mathbb{P}(G) = 1- e^{-\hat{\lambda} \pi \hat{r}_{\gamma}^2},$$
$$ \mathbb{P}(H) = e^{-\hat{\lambda} \pi \hat{r}_{\gamma}^2} (1- e^{-\lambda \pi r_{\gamma}^2}).$$

Let $\hat{J}_1^{(\gamma)}$ and  $\hat{J}_2^{(\gamma)}$ denote the area similar to that of what we considered before.  Given the mobile is served by macro BS, we need to consider the users which are within the area $\hat{J}_1^{(\gamma)}$  excluding the area covered by micro BS in this area.  The area covered by the micro BS in $\hat{J}_1^{(\gamma)}$  is approximated to be $\nu^{(\hat{r}_{\gamma})} \times \mathbb{E}[\hat{J}_1^{\gamma}]$, where $\nu^{(\hat{r}_{\gamma})}$ is the volume fraction associated with micro-BS as given by (\ref{eq: rho with r}). Then

$$ \mathbb{E}\big[\frac{1}{N_p^{(\gamma)}+1} |G\big]  = \frac{1-e^{\xi \mathbb{E}[\hat{J}_2^{(\gamma)}]}}{\xi \mathbb{E}[\hat{J}_2^{(\gamma)}]}, $$
$$ \mathbb{E}\big[\frac{1}{N_p^{(\gamma)}+1} | H\big]  = \frac{1-e^{\xi \mathbb{E}[\hat{J}_1^{(\gamma)}](\hat{\lambda} \pi \hat{r}_{\gamma}^2 )}}{\xi \mathbb{E}[\hat{J}_1^{(\gamma)}]( \hat{\lambda} \pi \hat{r}_{\gamma}^2 )}. $$

For a given density of macro-BS $\lambda$, the density of micro-BS $\hat{\lambda}$ and the density of users $\xi$, we can evaluate the SNR value $\gamma^*$ for which the load factor $\rho(\gamma,\xi)$ given in (\ref{eqrx1}) is maximum. Fig.\ref{lambda2_gamma}  illustrates the optimal SNR value ($\gamma^*$) with varying density of micro BS $\hat{\lambda}$.  Notice that the optimal gamma value initially decreases with the density of micro-BS.

Since in our setting the micro BS have smaller transmission power, for a given threshold $\gamma$, the radius of coverage for micro BS is smaller than that of the macro BS .i.e., $\hat{r_{\gamma}} < r_{\gamma}$.  Since the micro-BS are given higher preference, with an increase in their density, the optimal threshold decreases in order to increase their coverage.  Also, the cost incurred by increasing coverage of macro-BS is compensated by the increased density and coverage of micro-BS till a certain density.

\begin{figure}[h]
\centering
\includegraphics[scale=0.5]{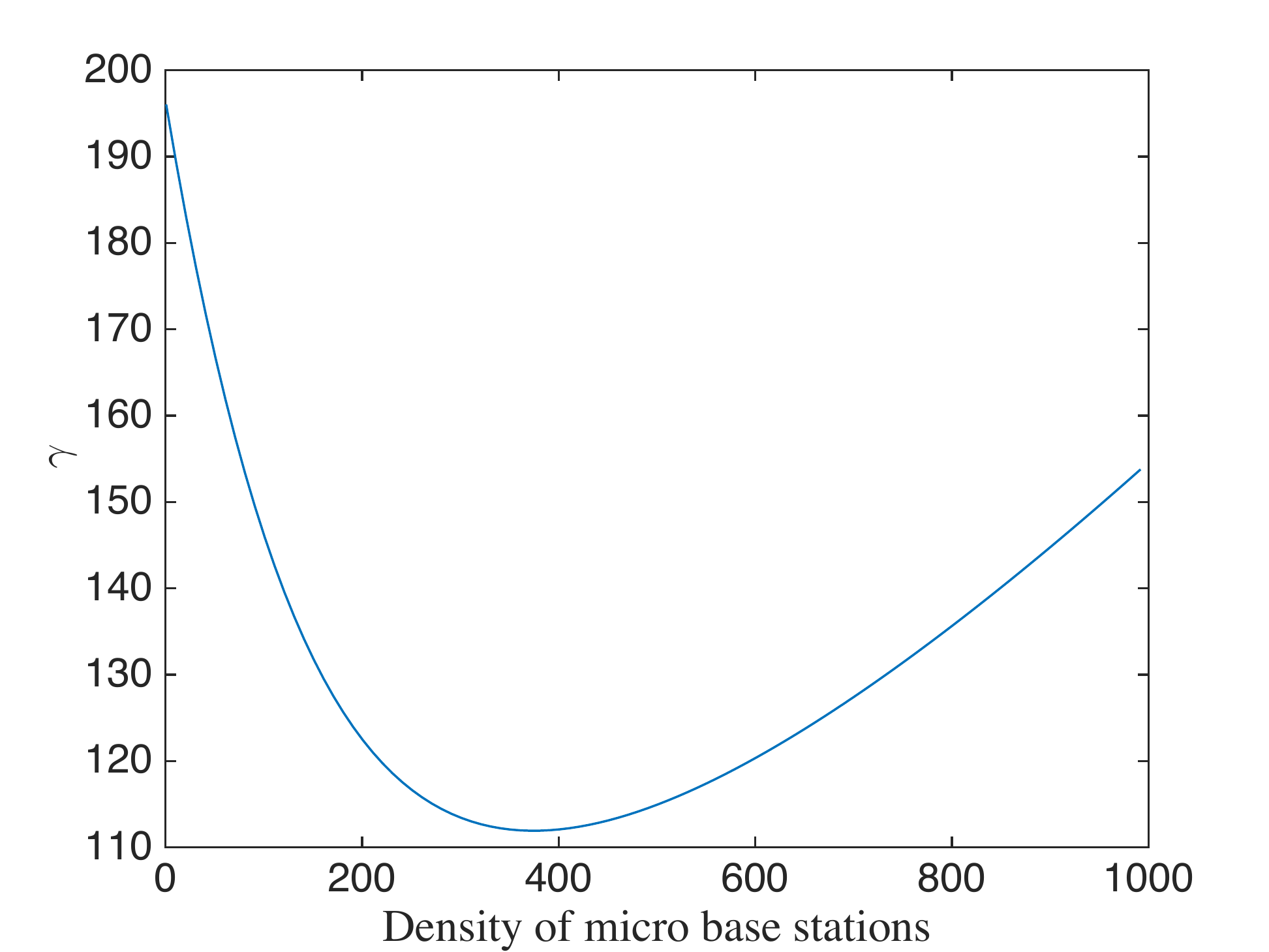}
\caption{ $\gamma^*$ with increase in micro- BS density for a certain fixed density of macro-BS for an arbitrary
positive constants $a$ and $\eta$ in heterogeneous case.
}
\label{lambda2_gamma}
\end{figure}

Fig.\ref{lambda2_xi} illustrates the level set curve of $\rho(\gamma^*,\xi) =1$ for various values of $\hat{\lambda}$ and $\xi$ for a given density of macro-BS $\lambda$.  Given this setup, it is possible to answer questions like what is the minimum density of micro-base stations that needs to be deployed by the operator to serve a certain density of users, given the density of macro-BS $\lambda$ such that the video streaming is uninterrupted for all the users. Thus, operators can learn the cost incurred to serve a higher density of users by deploying micro-BS and in case the cost incurred is higher, operator might be interested in other technologies.

\begin{figure}[h]
\centering
\includegraphics[scale=0.5]{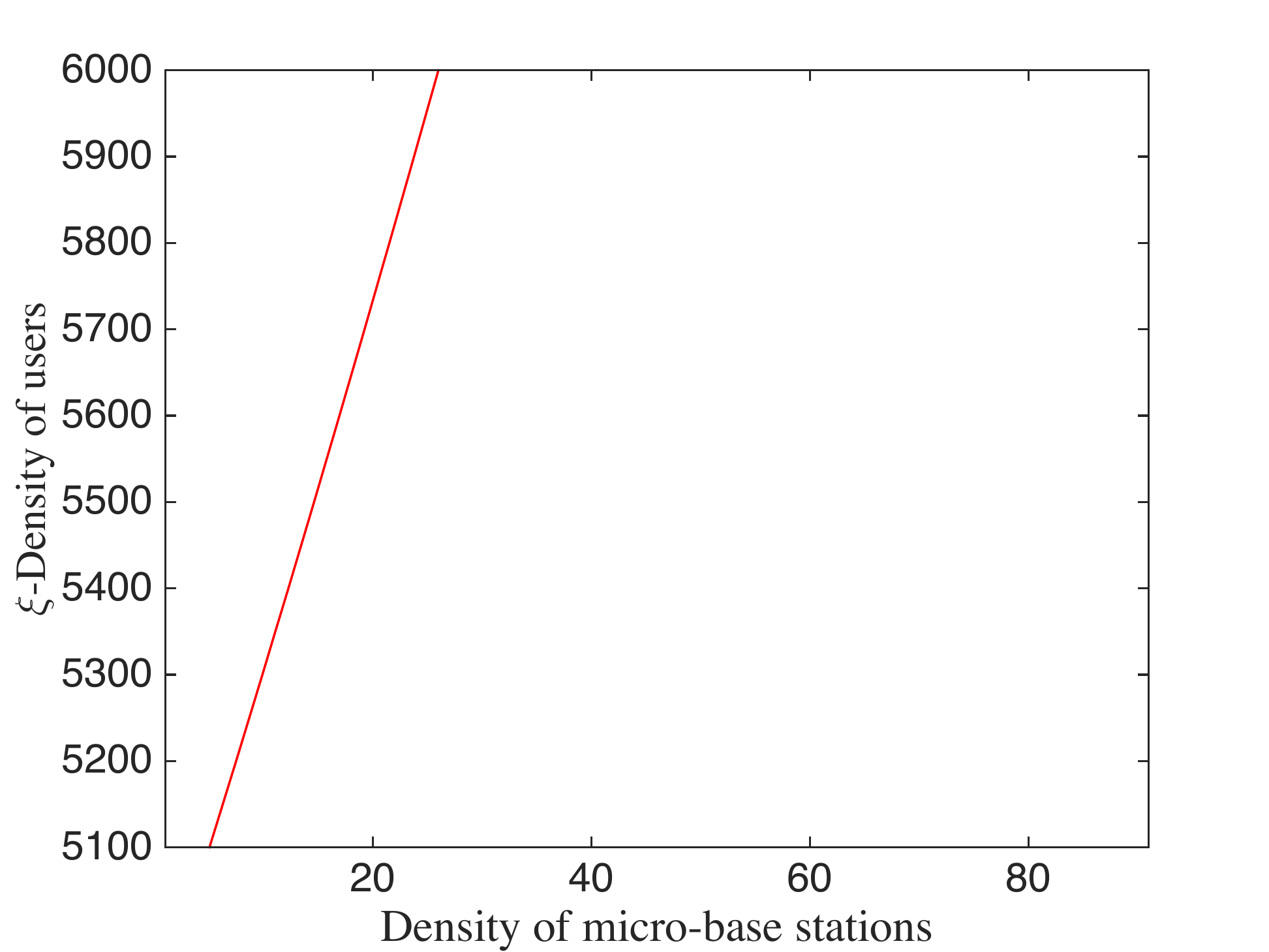}
\caption{Level set curve of $\rho(\gamma^*,\xi) =1$ for an arbitrary
positive constants $a$ and $\eta$ in heterogeneous case.
}
\label{lambda2_xi}
\end{figure}

\begin{remark}
In heterogeneous networks, the interference from macro-BS to micro-Bs is of major concern. One can introduce blanking, to reduce the effect of such interference i.e., assume that with certain probability $f$ micro-BS transmit and with probability $1-f$ macro BS transmit. \\

Also, by considering heterogeneous networks with cellular BS and Wi-Fi hotspots, there is no such problem of interference since both cellular BS and Wifi hotspots operate at different frequencies.

\end{remark}

\section{Conclusion}

As explained in the introduction, the analysis of temporal variations of the shared rate experienced by a mobile user requires the
characterization of both the functional distribution of a continuous
parameter stochastic process (rate process) constructed on a random spatial structure
(e.g. the Poisson Voronoi tessellation) and of another stochastic process (sharing number process) constructed on a random distribution of static users. This paper addressed the simplest question
of this type by focusing on the underlying SNR process and the sharing number process in the absence of fading.
This allowed us to derive an exact representation of the
level crossings of the stochastic process of interest as an
alternating renewal process with a full characterization
of the involved distributions and of the asymptotic behavior of rare events. 
The simplicity and the closed form nature of this mathematical picture
are probably the most important observations of the paper. We also showed by simulation that this very
simple model provides a good representation of the salient characteristics
which happen for more complex systems, like e.g. 
those with fading when the fading variance is small enough,
or those based on SINR rather than SNR when the path loss exponent
is large enough. This model is hence of potential practical use
as is, in addition to being a first glimpse at a set
of new research questions. The most challenging questions on the
mathematical side are probably
(1) the understanding of the tension between the randomness coming
from geometry (studied in the present paper), from sharing the network with other users (also studied in the present paper) and that coming from propagation
(only studied by simulation here): it would be nice to analytically
quantify when one dominates the other.
(2) the extension of the analysis to SINR processes, which are
our long term aim and will require
significantly more sophisticated mathematical tools (e.g. based on functional
distributions of shot noise fields), than those used so far.
On the practical side, the main future challenges are linked to the initial 
motivations of this work, namely in the prediction and optimization of the
user quality of experience. Many scenarios refining those studied can be considered. For instance, the stationary analysis of the
fluid queue representing video streaming should be completed by a transient
analysis and by a discrete time analysis.
This alone opens an interesting and apparently unexplored connection between
stochastic geometry and queuing theory with direct implications to
wireless quality of experience.

\section{Acknowledgments}

This work is supported in part by the National Science Foundation under Grant
No. NSF-CCF-1218338, NSF Grant CNS-1343383 and an award from the Simons Foundation (No.197982), all
to the University of Texas at Austin.

\IEEEpeerreviewmaketitle

\bibliographystyle{IEEEtran}
\bibliography{references.bib}{}

\section*{Appendix}

\subsection{Proof of Theorem \ref{th:thm1}}

The SNR process is closely related to the process tracking the
distance of the mobile user to the closest node (see Eq. \eqref{eq:snr}). Let us consider the projection $z$ of the mobile user's closest node $Z$ onto the straight line. Then, $z$ is a point of $\chi$ if $z$ lies within the Voronoi cell of $Z$. Thus, the point process can be characterized by the fact that given a node at a height $h$ from the straight line, the closed ball $B_h(z)$ of radius $h$ centered at the projection point is empty of nodes. Consider a rectangle of unit length and height $2x$ such that the straight line passes through the center. Thus, from Mecke's formula \cite{BB3}, the intensity of the point process $\chi$ is
\begin{equation}
\label{rateofextrema}
\begin{split}
\mu_{\chi} &= v \lim_{x \to \infty} \lambda 2x \bigg( \int_{-x}^{x} e^{-\lambda \pi h^2} \frac{1}{2x} dh \bigg)\\ \\
&= v \lim_{x \to \infty} \sqrt{\lambda} \mbox{erf}(\sqrt{\lambda \pi} x) = v \sqrt{\lambda}.
\end{split}
\end{equation}

\begin{figure}[!t]
\centering
\includegraphics[scale=0.4]{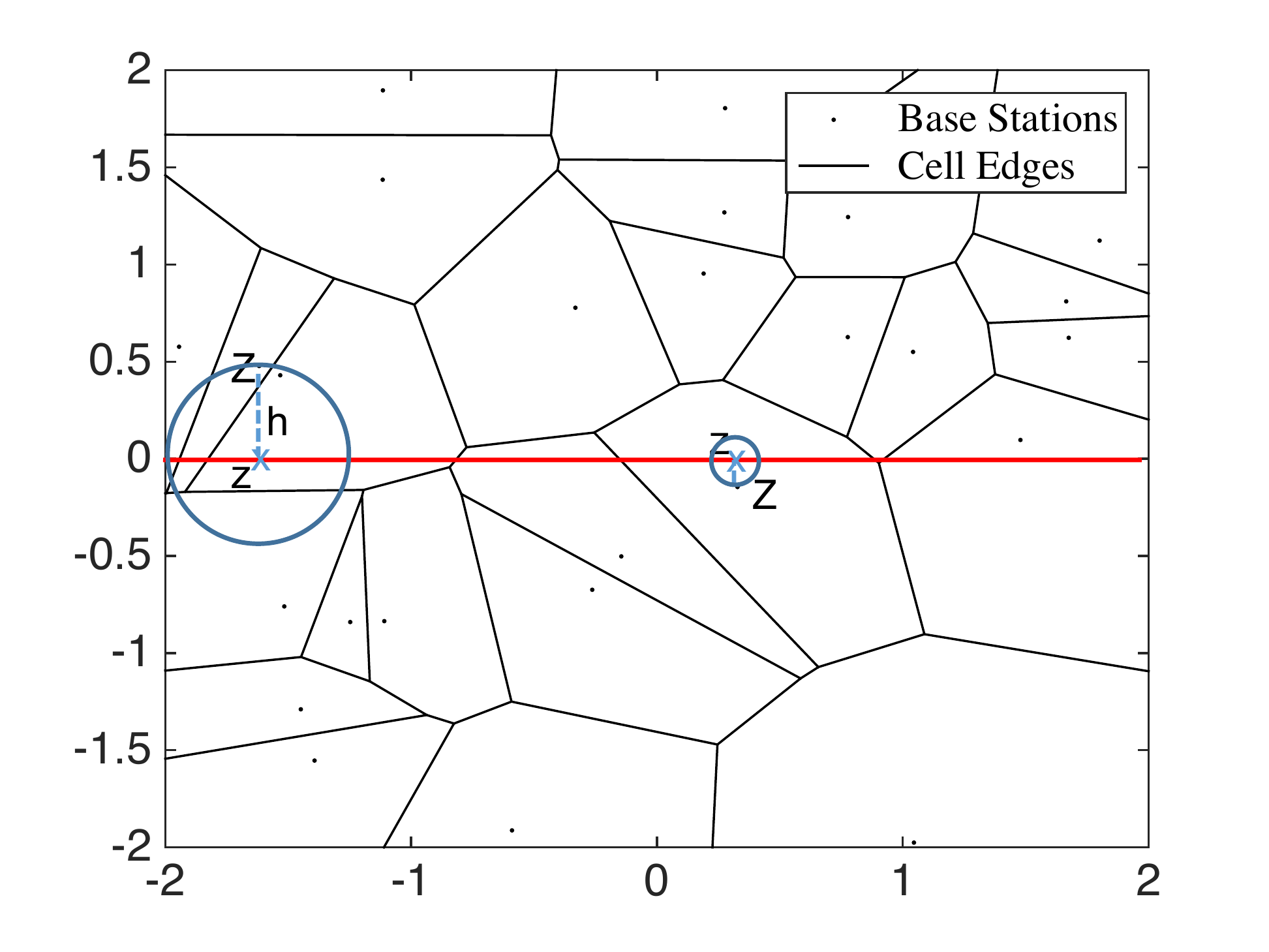}
\caption{The disc $B_z(h)$ around the projection point $z$ in two scenarios: where $z$ belong to the Voronoi cell of nucleus $Z$ and where it doesn't.}
\label{disc around mobile}
\end{figure}

\subsection{Busy Period of the M/GI/$\infty$ Queue}

\begin{theorem}(Makowski \cite{Makowski})
\label{makowski}
Consider an $M/GI/\infty$ queue with arrival rate $\lambda$ and generic service time $W$.
Let $M$ denote an $\mathbb{N}$-valued random variable which is geometrically distributed according to
\begin{equation}
\label{eq:geometric dist}
\mathbb{P}(M = l) = (1-\nu)(\nu)^{l-1}, l = 1,2,\ldots
\end{equation} 
with 
\begin{equation}
\label{eq: rho}
\nu = 1 - e^{-\rho}   ~\mbox{and} ~ \rho = \lambda \mathbb{E}[W].
\end{equation}
Consider the $\mathbb{R^{+}}$-valued random variable $U$ distributed according to
\begin{equation}
\label{eq: rv U}
\mathbb{P}(U \leq u) = \frac{1}{\nu} (1 - e^{-\rho \mathbb{P}[\hat{W} \leq u]}), u  \geq 0,
\end{equation}
where, 
$\hat{W}$ is the forward recurrence time associated with the generic service time $W$.
Let $\{ U_n , n \geq 1\}$ be an i.i.d. sequence independent of the random variable $M$. 
Let $B$ denote a typical busy period. Then the forward recurrence time $\hat{B}$ associated with $B$ admits the following random sum representation:
\begin{equation}
\label{eq:random sum}
\hat{B} =_{d} \sum_{i=1}^{M} U_i,
\end{equation}
where $=_{d}$ denotes equality in distribution.
\end{theorem}

\subsection{Proof of Theorem \ref{asymptotic}}
$V^{({\gamma})}$ is the sum of the busy period $B^{({\gamma})}$ and the idle period $I^{({\gamma})}$ of the $M/GI/\infty$ queue. We know that the distribution of the idle period is exponential with parameter $2 \lambda v r_{\gamma}$ and that the busy period $B^{({\gamma})}$  and idle period $I^{({\gamma})}$ are independent.

The Laplace transform of the random variable $V^{({\gamma})}$ scaled by $2 \lambda v r_{\gamma}$ is given by

\begin{equation}
\label{eq:laplace interarrival}
\Phi_{V^{(\gamma)}}(f(\gamma)s) = \Phi_{B^{({\gamma})}}(f(\gamma) s) \Phi_{I^{({\gamma})}}(f(\gamma) s),
\end{equation}
\mbox{where},
\begin{equation}
\label{eq: laplace idle}
\Phi_{I^{({\gamma})}}(f(\gamma) s)  = \frac{1}{1+ s}. 
\end{equation}

So, asymptotically, as $r_{\gamma}$ goes to 0 , $ \Phi_{I^{({\gamma})}}(f(\gamma) s) $ converges in distribution to an exponential random variable with unit mean. 

Next, we need to prove that asymptotically, the busy period $B^{({\gamma})}$ scaled by $f(\gamma)$  goes to one. Let us consider the forward recurrence time of the busy period, $\hat{B}^{({\gamma})}$ which is defined in \eqref{eq:forward recurrence}. Now from Theorem ~\ref{level crossing characterization} we get,

$$
\mathbb{E}[e^{-s\hat{B}^{({\gamma})}}] = \mathbb{E}[e^{-s \sum_{i=1}^{M^{({\gamma})}} U^{({\gamma})}_i}] = \frac{(1-\nu^{({\gamma})})\mathbb{E}[e^{-sU^{({\gamma})}}]}{1-\nu^{({\gamma})}\mathbb{E}[e^{-sU^{({\gamma})}}]} .
$$
and it follows that
\begin{equation}
\label{eq:laplace busy}
\mathbb{E}[e^{-s f(\gamma) \hat{B}^{({\gamma})}}] = \frac{(1-\nu^{({\gamma})})\mathbb{E}[e^{-s f(\gamma) U^{({\gamma})}}]}{1-\nu^{({\gamma})}\mathbb{E}[e^{-s f(\gamma) U^{({\gamma})}}]} .
\end{equation}

Now, 
\begin{equation}
\label{eq:laplace U}
 \mathbb{E}[e^{-s f(\gamma) U^{({\gamma})}}]  =  \sum_{k=0}^{\infty} l_k(s,r_{\gamma}).
\end{equation}
with
$$ l_k(s,r_{\gamma}) = \frac{ (-1)^k (2 \lambda v sr_{\gamma})^k \mathbb{E}[(U^{({\gamma})})^k]}{k!}.  $$
 
Given that the support of service times $W^{({\gamma})}$ is $[0,2r_{\gamma}/v]$, it follows from \eqref{eq: rv U} that the support of the random variable $U^{({\gamma})}$ is also $[0,2r_{\gamma}/v]$ and all its moments are bounded $\mathbb{E}[(U^{({\gamma})})^k] < (2r_{\gamma}/v)^k$. Thus, for all values of $s$, we get that $\lim_{r_{\gamma} \to 0} l_k(s,r_{\gamma}) =0.$
 
For all $k = 2, 3..$ we have that
$$ |l_k(s,r_{\gamma})| \leq \frac{(2 \lambda v s r_{\gamma}^2)^k} {k!}.$$

 
Therefore, for all values of $s$,$\lim_{r_{\gamma} \to 0} \sum_{k=2}^{\infty} l_k(s,r_{\gamma}) = 0.$

Thus, Eq \eqref{eq:laplace U} can be written as 

 $$
 \mathbb{E}[e^{-s f(\gamma)U^{({\gamma})}}] = 1 - s 2 \lambda r_{\gamma} v \mathbb{E}[U^{({\gamma})}] +  o(r_{\gamma}),
 $$
Now, substituting into Eq \eqref{eq:laplace busy} we have  
\begin{equation}
\begin{split}
& \mathbb{E}[e^{-s f(\gamma) \hat{B}^{({\gamma})}}]  = \\ \\
 & \frac{1- \nu^{({\gamma})} +2s \lambda v r_{\gamma} \mathbb{E}[U^{({\gamma})}](1- \nu^{({\gamma})})+ o(r_{\gamma})}{1-\nu^{({\gamma})} + 2\nu^{({\gamma})} s \lambda v r_{\gamma} \mathbb{E}[U^{({\gamma})}]   + o(r_{\gamma}) }.
 \end{split}
\end{equation}

From which it follows that
\begin{equation}
\label{eq:limit1}
\lim_{r_{\gamma} \to 0} \mathbb{E}[e^{-s f(\gamma) \hat{B}^{({\gamma})}}] = 1.
\end{equation}

From ~\cite{Makowski} we have that 
\begin{equation}
\label{eq: laplace relation}
\mathbb{E}[e^{-s f(\gamma) B^{({\gamma})}}] = 1 - sf(\gamma) \mathbb{E}[B^{({\gamma})}] \mathbb{E}[e^{-s f(\gamma) \hat{B}^{({\gamma})}}],
\end{equation}
and also
\begin{equation}
\label{eq:mean busy}
\mathbb{E}[B^{({\gamma})}] = \frac{\nu^{({\gamma})}}{2\lambda vr_{\gamma} (1-\nu^{({\gamma})})}.
\end{equation}
Thus, we get

$$
\mathbb{E}[e^{-s f(\gamma) B^{({\gamma})}}]  = 1 - s\frac{1-e^{-\lambda \pi r_{\gamma}^2}}{e^{-\lambda \pi r_{\gamma}^2}} \mathbb{E}[e^{-s f(\gamma) \hat{B}^{({\gamma})}}].
$$
Using the limit in \eqref{eq:limit1}, we get

\begin{equation}
\label{eq:result1}
\lim_{r_{\gamma} \to 0} \mathbb{E}[e^{-s f(\gamma) B^{({\gamma})}}] = 1.
\end{equation}

Thus, the distribution of the random variable $V^{({\gamma})}$ scaled by  $f(\gamma) = 2 \lambda v r_{\gamma} $ converges in distribution to an exponential random variable with unit mean.

Now, consider a continuous function $g(\gamma)$ such that $\lim_{r_{\gamma} \to \infty} g(\gamma) = 0$. The Laplace transform of the random variable $V^{({\gamma})}$ scaled by $g(\gamma)$ is given by

\begin{equation}
\label{eq: laplace interarrival2}
\Phi_{V^{({\gamma})}(s)} = \Phi_{B^{({\gamma})}}(g(\gamma) s) \Phi_{I^{({\gamma})}}(g(\gamma) s),
\end{equation}
\mbox{where},

\begin{equation}
\label{eq: laplace idle2}
\Phi_{I^{({\gamma})}}(g(\gamma) s)  = \frac{1}{1+ s g(\gamma)  (1 / 2 vr_{\gamma} \lambda )}. 
\end{equation}

Asymptotically, as $r_{\gamma}$ goes to $\infty$, $ \Phi_{I^{({\gamma})}}(g(\gamma) s) $ goes to 1. 

By the same arguments as those for the up-crossing case
$$
\mathbb{E}[e^{-s g(\gamma)\hat{B}^{({\gamma})}}] = \frac{(1-\nu^{({\gamma})})\mathbb{E}[e^{-sg(\gamma)U^{({\gamma})}}]}{1-\nu^{({\gamma})}\mathbb{E}[e^{-sg(\gamma)U^{({\gamma})}}]},
$$
where,
\begin{equation}
\label{eq:laplace U2}
\mathbb{E}[e^{-sg(\gamma)U^{({\gamma})}}]  =  \sum_{k=0}^{\infty} \frac{(-1)^k( 2\lambda v sr_{\gamma})^k (e^{-\rho^{({\gamma})}})^{k}}{k!}  \mathbb{E}[(U^{({\gamma})})^k].
\raisetag{1\baselineskip}
\end{equation}

Let $m_k(s,r_{\gamma}) = \frac{(-1^k)(2\lambda v sr_{\gamma})^k (e^{-\rho^{({\gamma})}})^{k-1}}{k!}  \mathbb{E}[(U^{({\gamma})})^k]. $
 The moments of the random variable $U^{({\gamma})}$ are bounded $\mathbb{E}[(U^{({\gamma})})^k] < (2r_{\gamma}/v)^k$. Thus, for all values of $s$, we get that $\lim_{r_{\gamma} \to \infty} m_k(s,r_{\gamma}) =0,$ because $\rho^{({\gamma})} = \lambda \pi r_{\gamma}^2$ and the exponential term $(e^{-\rho^{({\gamma})}})^{k-1}$ dominates. Also for $k =2,3,...$ we have that

$$|m_k(s,r_{\gamma})| \leq \frac{(2 \lambda v s r_{\gamma}^2)^k (e^{-\rho^{({\gamma})}})^{k-1}} {k!}.$$

 
Therefore for all values of $s$,
$$
\lim_{r_{\gamma} \to \infty} \sum_{k=2}^{\infty} m_k(s,r_{\gamma}) = 0.
$$

Thus, Eq \eqref{eq:laplace U2} can be written as 
 $$
 \mathbb{E}[e^{-sf(\gamma)U^{({\gamma})}}] = 1 - s e^{-\rho^{({\gamma})}}2vr_{\gamma}\lambda  \mathbb{E}[U^{({\gamma})}] + e^{-\rho^{({\gamma})}} o(r_{\gamma}),
 $$
where, $ \lim_{r_{\gamma} \to \infty} o(r_{\gamma}) = 0$.

From Lemma \ref{equality} we have that,$
\lim_{r_{\gamma} \to \infty} 2 \lambda r_{\gamma}v \mathbb{E}[U^{({\gamma})}] = 1.$
Putting these results together we have that
 \begin{equation}
 \begin{split}
& \mathbb{E}[e^{-s g(\gamma)\hat{B}^{({\gamma})}}] = \\ \\
& \frac{ e^{-\rho^{({\gamma})}}  - s e^{-2\rho^{({\gamma})}}2vr_{\gamma}\lambda  \mathbb{E}[U^{({\gamma})}]  + e^{-2\rho^{({\gamma})}} o(r_{\gamma})}{e^{-\rho^{({\gamma})}}\big[1-e^{-\rho^{({\gamma})}}\big]\big[s 2vr_{\gamma}\lambda  \mathbb{E}[U^{({\gamma})}] - e^{-\rho^{({\gamma})}} o(r_{\gamma})\big] +  e^{-\rho^{({\gamma})}}}
 \end{split}
 \end{equation}
so that

\begin{equation}
\label{eq:limit2}
\lim_{r_{\gamma} \to \infty} \mathbb{E}[e^{-s g(\gamma)\hat{B}^{({\gamma})}}] = \frac{1}{1+s}.
\end{equation}

Now, from Eq \eqref{eq: laplace relation} and \eqref{eq:mean busy} we get,
$$
\mathbb{E}[e^{-s f(\gamma) B^{({\gamma})}}]  = 1 - s (1-e^{-\lambda \pi r_{\gamma}^2}) \mathbb{E}[e^{-s f(\gamma) \hat{B}^{({\gamma})}}].
$$
Thus from the limit in \eqref{eq:limit2},$
\lim_{r_{\gamma} \to \infty} \mathbb{E}[e^{-s g(\gamma) B^{({\gamma})}}] = \frac{1}{1+s}.$

Thus, the random variable $V^{({\gamma})}$ scaled by  $g(\gamma) = 2 \lambda v r_{\gamma} e^{-\lambda \pi r_{\gamma}^2} $ converges in distribution to an exponential random variable with unit mean.

\subsection{Proof of Theorem \ref{thhsr}}

Let us first prove the second relation.
Let $K=p/w$ and let $D$ denote the distance to the closest base
station.  We have
\begin{eqnarray*}
& & \hspace{-1cm}\mathbb{P}((\log( 1+{\mathrm{SNR}}) > s) \\
& = & \mathbb{P}\left(
KD^{-\beta} > e^s -1\right)\\ 
& = & \mathbb{P}\left( D^{\beta} < \frac K{e^s -1} \right)\\
& = & \mathbb{P}\left( D^{2} < \left(\frac K {e^s -1} \right)^{\frac 2 \beta}
\right)\\
& = & 1 -
\exp\left(- \lambda \pi \left(\frac K {e^s -1} \right)^{\frac 2 \beta} \right),
\end{eqnarray*}
where we used the fact that $D^2$ is an exponential random variable
with parameter $\lambda \pi$. The result then follows from 
the bound $ a \ge 1-e^{-a} \ge a + a^2/2$, for $a\ge 0$. 

We now prove the first relation. Since
$\hat S^{(\gamma)} \le S^{(\gamma)} \le \log( 1+{\mathrm{SNR}})$, in order to prove
the first inequality, it is enough to show that
\begin{eqnarray}
\label{eq:lbb}
& &  \lim_{s\to \infty} - \frac 1 s \log\left(
\mathbb{P}(\hat S^{(\gamma)} > s)\right) 
\le \frac 2 \beta.
\end{eqnarray}
We have
\begin{eqnarray*}
\mathbb{P}(\hat S^{(\gamma)} > s) & = &  \mathbb{P}( \log(1+KD^{-\beta})>s(1+ \hat N_{\gamma})).
\end{eqnarray*}
Hence
\begin{eqnarray*}
& &\hspace{-1cm}
 \mathbb{P}(\hat S^{(\gamma)} > s) \\
& = & \mathbb{P}\left(
KD^{-\beta} > \exp\left(s(1+ \hat N_{\gamma})\right) -1\right)\\ 
& \ge & \mathbb{P}\left( KD^{-\beta} > \exp\left(s(1+ \hat N_{\gamma})\right) \right)\\
& = & \mathbb{P}\left( D^{\beta} < K \exp\left(-s(1+ \hat N_{\gamma})\right) \right)\\
& = & \mathbb{P}\left( D^{2} < K^{\frac 2 \beta}
\exp\left(-s{\frac 2 \beta}(1+ \hat N_{\gamma})\right) \right)\\
& = & \mathbb{P}\left( D^{2} < K^{\frac 2 \beta}
e^{-s{\frac 2 \beta}} e^{-s{\frac 2 \beta}\hat N_{\gamma}} \right).
\end{eqnarray*}
Using now the fact that $D^2$ is an exponential random variable
with parameter $\lambda \pi$, independent of $\hat N_{\gamma}$, we get
\begin{eqnarray*}
& & \mathbb{P}(\hat S^{(\gamma)} > s) \\
& \ge & 1- \mathbb{E}\left( \exp\left(-\lambda \pi K^{\frac 2 \beta}
e^{-s{\frac 2 \beta}} e^{-s{\frac 2 \beta} \hat N_{\gamma}}\right)\right)\\
& \ge & \lambda \pi K^{\frac 2 \beta}
e^{-s{\frac 2 \beta}}
\exp\left(-\xi \pi r_{\gamma}^2 \left(1-e^{-s{\frac 2 \beta}}\right)\right)
\left(1+ o\left(e^{-s{\frac 2 \beta}}\right)\right),
\end{eqnarray*}
where we used again the bound $1-e^{-a} \ge a + a^2/2$ and the fact
that $\hat{N}^{(\gamma)}$ is Poisson with parameter $\xi \pi r_{\gamma}^2$.

Taking the log, multiplying both sides by $-\frac 1 s$ and letting
$s$ go to infinity gives (\ref{eq:lbb})

\subsection{Proof of Theorem \ref{thlsr}}

Let us first prove the last relation. Using Chernoff's bound
for the Poisson random variable $\hat N^{(\gamma)}$, we get
\begin{eqnarray*}
& & \hspace{-1cm}
\mathbb{P}\left( \hat N^{(\gamma)} +1> s_n \log (1 +Kr_{\gamma}^{-\beta}) \right)\\
& \le & e^{-\theta} \left( \frac {e \theta} {s_n \log (1 +Kr_{\gamma}^{-\beta})}
\right)^{s \log (1 +Kr_{\gamma}^{-\beta})},
\end{eqnarray*}
with $\theta= \xi \pi r_{\gamma}^2$, which shows that
\begin{eqnarray*}
&& \hspace{-1cm} \lim_{n\to \infty} - \frac 1 {s_n \log(s_n)} \log\left(
\mathbb{P}\left( \hat N^{(\gamma)} +1> s_n \log (1 +Kr_{\gamma}^{-\beta}) \right)\right)
\\ & \ge & \log (1 +Kr_{\gamma}^{-\beta}). 
\end{eqnarray*}
The converse inequality follows from the bound
\begin{eqnarray*}
\mathbb{P}\left( \hat N^{(\gamma)} +1> s_n \log (1 +Kr_{\gamma}^{-\beta}) \right)
\ge 
e^{-\theta} \frac {\theta^k}{k!}, \end{eqnarray*}
with $k =   s_n \log (1 +Kr_{\gamma}^{-\beta}) $
and Stirling's bound on the Gamma function.

Let us now prove the first relation.
We have
\begin{eqnarray*}
& & \hspace{-1cm}
\mathbb{P}\left( \hat S^{(\gamma)} < \frac 1 s_n  \right)\\
& = &
\int_0^{r^2_\gamma} e^{-\lambda \pi v}
\mathbb{P}\left( \hat N^{(\gamma)} + 1 > s_n \log\left(1+Kv^{-\frac \beta 2}\right) \right)
{\mathrm d} v.
\end{eqnarray*}
For $s_n$ large enough,
\begin{eqnarray*}
& & \hspace{-1cm}
\mathbb{P}\left( \hat S^{(\gamma)} < \frac 1 s_n  \right)\\
& \le & 
\mathbb{P}\left( \hat N^{(\gamma)} + 1 > s_n \log\left(1+Kr_{\gamma}^{-\frac \beta 2}\right) \right)
\int_0^{r^2_\gamma} e^{-\lambda \pi v}
{\mathrm d} v,
\end{eqnarray*}
which allows one to conclude that
\begin{eqnarray*}
& & \hspace{-1cm} \lim_{n\to \infty} - \frac 1 {s_n \log(s_n)}
\log\left( \mathbb{P}\left(\hat S^{(\gamma)} < \frac 1 s_n\right)\right)
\ge  \log (1 +Kr_{\gamma}^{-\beta}).
\end{eqnarray*}

The upper bound follows from the inequality:

\begin{eqnarray*}
& & \hspace{-1cm}
\mathbb{P}\left( \hat S^{(\gamma)} < \frac 1 s_n  \right)\\
& = &
\sum_{l=0}^{\infty} 
e^{-\theta} \frac{\theta^l}{l!} 
 \mathbb{P}\left( D^{2} > \left(\frac K {e^{\frac{l+1}{ s }} -1} \right)^{\frac 2 \beta}
\right)\\
& \ge & 
e^{-\theta} \frac{\theta^k}{k!} 
 \mathbb{P}\left( D^{2} > \left(\frac K {e^{\frac{k+1}{s} } -1} \right)^{\frac 2 \beta}
\right)
\end{eqnarray*}
where $k =   s_n \log (1 +Kr_{\gamma}^{-\beta}) $.

\subsection{Integral Geometry \RN{1} }

Consider a tagged user at the origin, and let the serving BS $X_j$ be at a distance $x$ from the origin. Let $Q$ denote a point a distance $u$ from the origin. We need to consider all points $Q$ that are within a distance $r$ from the serving BS $X_j$ i.e., $z = \left | X_j Q \right | = \sqrt{u^2+x^2-2ux\mbox{cos}(\alpha)} < r $, where $\alpha = \angle{X_jOQ}$ as shown in Figure 8.

\begin{figure}
\centering
\includegraphics[scale=0.5]{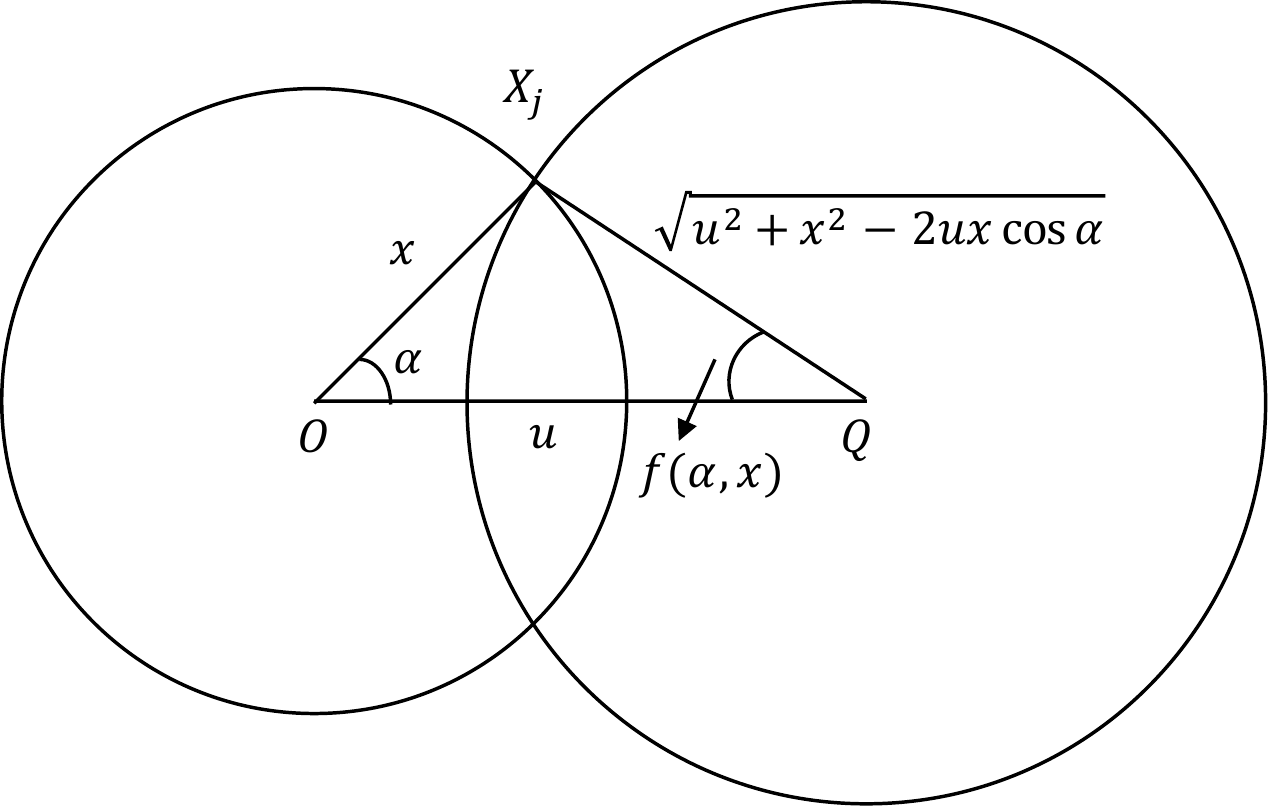}
\label{intgeometry}
\caption{Tagged user at origin, its serving base station $X_j$ at distance $x$ and a point $Q$ at a distance $u$ along with discs $B_x(0)$ and $B_z(Q)$.}
\end{figure}
Let $B_x(0),B_z(Q)$ be two discs with centers at the origin and $Q$, and radii $x$ and $X_jQ$, respectively.
 
The conditional probability that a point at distance $u$ from origin is within the Voronoi cell of the BS serving the user at the origin given it is at a distance $x$ is the probability that there is no other BS in the area of disc $B_z(Q)$ excluding the area of $B_x(0)$. Then the conditional expectation is:
 \begin{equation}
 \begin{split}
&\mathbb{E}[V^*| D= x] = \\ \\
& 2 \int_0^{\pi} \int_0^{x cos(\alpha) + \sqrt{r^2 - x^2sin(\alpha)}} e^{(-\lambda l(B_z(Q) - B_x(0)))} u dud\alpha ,
\end{split}
\raisetag{3\baselineskip}
\end{equation}
where, $D$ is a random variable denoting the distance to the closest BS.Thus,

\begin{equation}
\begin{split}
\mathbb{E}[V^*] &= \int_0^r \mathbb{E}[V^*| D= x]  \frac{f_D(x)}{\int_0^r f_D(y) dy} dx\\
&= \frac{\int_0^r \mathbb{E}[V^*| D= x] 2\lambda \pi x e^{-\lambda \pi x^2} dx }{1 - e^{-\lambda \pi r^2}}\\
&= \frac{1}{1 - e^{-\lambda \pi r^2}}\int_0^r 2 \int_0^{\pi} \int_0^{x \cos(\alpha) + \sqrt{r^2 - x^2\sin(\alpha)}} \\ \\
&~~~~ e^{(-\lambda l(B_z(Q) - B_x(0)))} u dud\alpha  2\lambda \pi x e^{-\lambda \pi x^2} dx \\
&= \frac{1}{1 - e^{-\lambda \pi r^2}} 4 \lambda \pi  \int_0^{r} \int_0^{\pi} \int_0^{x \cos(\alpha) + \sqrt{r^2 - x^2\sin(\alpha)}} \\\\
&~~~~e^{(-\lambda l(B_x(0) \cup B_z(Q))) }u x dud\alpha dx.
\end{split}
\raisetag{1\baselineskip}
\end{equation}

The last equality is because $l(B_x(0)) = \pi r^2$ and 
$l(B_x(0) \cup B_z(Q)) = ux\sin(\alpha) + (\pi-\alpha)x^2 + (\pi - f(\alpha,x))(u^2+x^2-2ux\cos(\alpha))$ with $\cos(f(\alpha,x)) = \frac{u-x\cos(\alpha)}{\sqrt{u^2+x^2-2ux\cos(\alpha)}}$.

\subsection{Integral Geometry \RN{2}}

 Let us consider a point $U$ at a distance $u$ from the origin on the line  $d(0,\theta)$ as shown in the figure. Then, the probability that the point is in the Voronoi cell of the BS $X_j$ is given by $e^{- \lambda l(C(u,\theta))}$, where $l(C(u,\theta))$ is the area of the shaded region in the figure. The conditional expectation is:
 \begin{equation} 
 \begin{split}
&\mathbb{E}[l(0,\theta) | D = x] =\\ \\
 & \frac{1}{\pi} \int_0^{\pi} \int_0^{x \cos(\theta) + \sqrt{r^2 - x^2\sin(\theta)}} e^{(-\lambda l(C(u,\theta)))} dud\theta,
\end{split}
\end{equation}
where, $D$ is the random variable denoting the distance to the closest BS.

Let $B_x(0),B_z(U)$ be two discs with centers at the origin and $U$, and radii $x$ and $X_jU$, respectively. Thus,

\begin{equation}
\begin{split}
\mathbb{E}[l(0,\theta)] &= \int_0^r \mathbb{E}[l(0,\theta)| D= x]  \frac{f_R(x)}{\int_0^r f_D(y) dy} dx\\
&= \frac{\int_0^r \mathbb{E}[l(0,\theta)| D= x] 2\lambda \pi x e^{-\lambda \pi x^2} dx }{1 - e^{-\lambda \pi r^2}}\\
&= \frac{1}{\pi(1 - e^{-\lambda \pi r^2})}\int_0^r \int_0^{\pi} \int_0^{x cos(\theta) + \sqrt{r^2 - x^2sin(\theta)}} \\ \\
&~~~~ e^{(-\lambda l(C(u,\theta)))} dud\theta  2\lambda \pi x e^{-\lambda \pi x^2} dx \\
&= \frac{1}{\pi(1 - e^{-\lambda \pi r^2})} 2 \lambda \pi  \int_0^{r} \int_0^{\pi} \int_0^{x cos(\theta) + \sqrt{r^2 - x^2sin(\theta)}} \\\\
&~~~~e^{(-\lambda  l(B_x(0) \cup B_z(U))) } x dud\theta dx.
\raisetag{1\baselineskip}
\end{split}
\end{equation}
The last equality is because $l(B_x(0))  = \pi r^2$ and 
$l(B_x(0) \cup B_z(U)) = ux\sin(\theta) + (\pi-\theta)x^2 + (\pi - f(\theta,x))(u^2+x^2-2uxcos(\theta))$ with $\cos(f(\theta,x)) = \frac{u-x\cos(\theta)}{\sqrt{u^2+x^2-2ux\cos(\theta)}}$.

\begin{figure}
\centering
\includegraphics[scale=0.5]{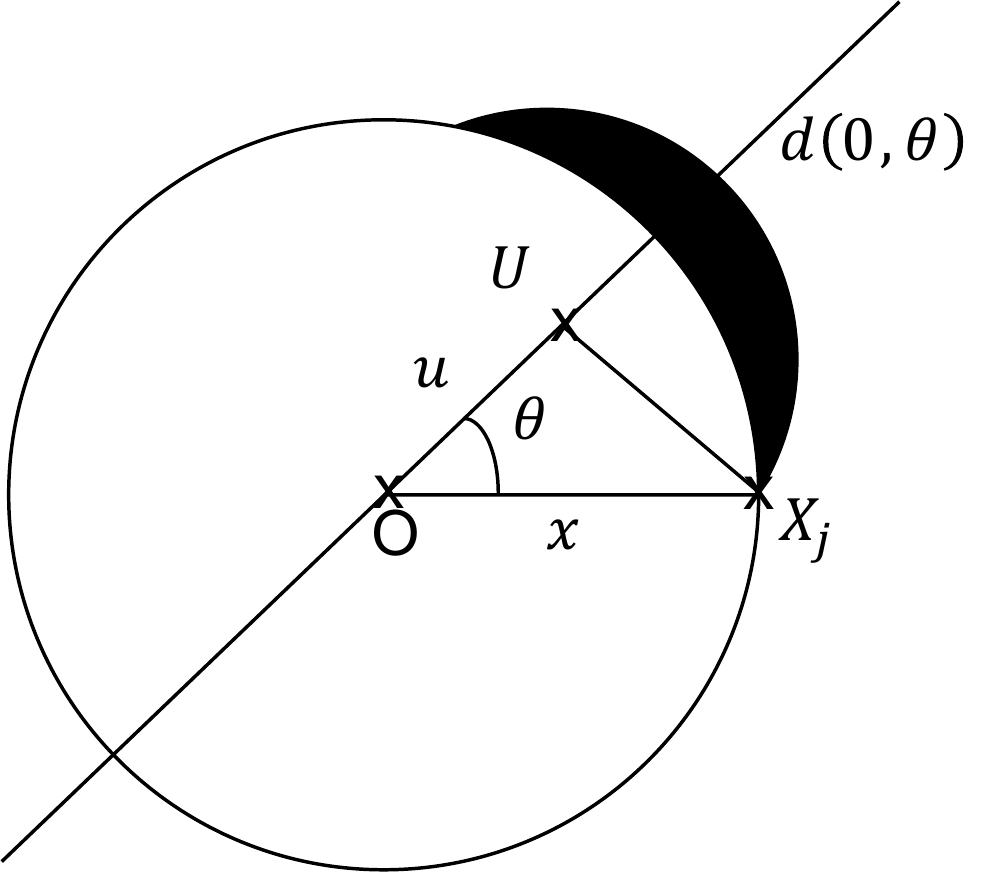}
\label{integralgeometry}
\caption{The tagged user at the origin, its serving base station $X_j$ at distance $x$ and a point $U$ at a distance $u$ on the line $d(0,\theta)$.}
\end{figure}

\subsection{Proof of Theorem 11}
\begin{proof}
Let $Z$ be a geometric random variable given by:
\begin{equation}
\label{eq:rv Z}
P(Z = l) = (1-\alpha)\alpha^{l}.
\end{equation}

Then, the random variable $T$ representing the time taken to download the file is the geometric sum of independent random variables given by
\begin{equation}
\label{eq:rv T}
T = \sum_{i=0}^{Z} (X_i + I^{(\gamma}) + Y.
\end{equation}
Thus, the Laplace transform of $T$ is
\begin{equation}
\label{eq:rv T}
\begin{split}
{\cal L}_{T}(s) &= \big[(1- \alpha) + (1-\alpha)\alpha {\cal L}_{X}(s){\cal L}_{I^{({\gamma})}}(s)+ ..\big] {\cal L}_{Y}(s) \\ \\
&= \frac{(1-\alpha){\cal L}_{Y}(s)}{1 -\alpha{\cal L}_{X}(s) {\cal L}_{I^{({\gamma})}}(s)} = \frac{(1-\alpha){\cal L}_{Y}(s)}{1 -\alpha{\cal L}_{X}(s) \frac {2\lambda v r}{2\lambda v r+ s}}.
\end{split}
\raisetag{2\baselineskip}
\end{equation}

The last equality follows from the fact that the $I^{({\gamma})}$ is exponential with parameter $2 \lambda v r$.
\end{proof}

\begin{lemma}
\label{equality}
For a $\mathbb{R^{+}}$ -valued random variable $U^{({\gamma})}$ defined by ~\eqref{eq: rv U} and service time distribution by ~\eqref{eq:service density}. We have the following limit:
$$
\lim_{r_{\gamma} \to \infty} 2 \lambda r_{\gamma}v \mathbb{E}[U^{({\gamma})}] = 1.
$$
\end{lemma}

From Equations \eqref{eq: rho with r} and \eqref{eq: rv U}:
\begin{equation}
\label{equation}
P[U^{({\gamma})}> u] = \frac{e^{-\rho^{({\gamma})}}} {\nu^{({\gamma})}} (e^{\rho^{({\gamma})}P[\hat{W}^{({\gamma})} > u]} - 1) , u \geq 0,
\end{equation}
and
\begin{equation}
\label{equation}
\begin{split}
\mathbb{E}[U^{({\gamma})}] & = \int_0^{2 r_{\gamma}/v} P[U^{({\gamma})} > u] du\\ \\
& = \frac{e^{-\rho^{({\gamma})}}} {\nu^{({\gamma})}} \int_0^{2 r_{\gamma}/v} (e^{\rho^{({\gamma})}P[\hat{W}^{({\gamma})} > u]} - 1) du.
\end{split}
\end{equation}

Thus,
\begin{equation}
\label{equation}
\mathbb{E}[U^{({\gamma})}]  = \frac{e^{-\rho^{({\gamma})}}} {\nu^{({\gamma})}} \bigg( \int_0^{2r/v} e^{\rho^{({\gamma})}P[\hat{W}^{({\gamma})} > u]} du  - \frac{2r_{\gamma}}{v} \bigg).
\end{equation}

Since  $\hat{W}^{({\gamma})}$ is the forward recurrence time associated with the service time $W^{({\gamma})}$ defined in ~\eqref{eq:forward recurrence},
$$
\mathbb{P}(\hat{W}^{({\gamma})} > u) = \frac{1}{\mathbb{E}[W^{({\gamma})}]} \int_u^{\infty} \mathbb{P} (W^{({\gamma})} > t) dt,
$$
where, $\mathbb{E}[W^{({\gamma})}] = \frac{ \rho^{({\gamma})}}{2r_{\gamma} v\lambda}$ from ~\eqref{eq: rho with r} and $\mathbb{P} (W^{({\gamma})} > t) = 1/2r_{\gamma} \sqrt{4r_{\gamma}^2 -  v^2 t^2}$ from ~\eqref{eq:service density}.

Thus,
\begin{equation}
\label{eq:integral}
\begin{split}
&\mathbb{E}[U^{({\gamma})}] = \frac{e^{-\rho^{({\gamma})}}} {\nu^{({\gamma})}} \bigg( \int_0^{\frac{2r_{\gamma}}{v}} e^{\lambda v \int_u^{\frac{2r_{\gamma}}{v}} \sqrt{4r_{\gamma}^2 - v^2 t^2} dt} du -  \frac{2r_{\gamma}}{v} \bigg)\\ \\
& = \frac{e^{-\rho^{({\gamma})}}} {\nu^{({\gamma})}} \bigg( \int_0^{\frac{2r_{\gamma}}{v}} e^{\lambda \int_u^{2r_{\gamma}} \sqrt{4r_{\gamma}^2 -  x^2} dx} du -  \frac{2r_{\gamma}}{v} \bigg) \\ \\
 & = \frac{1} {\nu^{({\gamma})}} \bigg(  \int_0^{\frac{2r_{\gamma}}{v}} e^{- \lambda ( \frac{1}{2} u \sqrt{4r_{\gamma}^2 - v^2u^2} + 2r_{\gamma}^2{\arctan}(\frac{uv}{\sqrt{4r_{\gamma}^2 - v^2u^2}})} du\\ \\
& ~~~~-  e^{-\lambda \pi r_{\gamma}^2}\frac{2r_{\gamma}}{v} \bigg)\\ \\
 & = \frac{1} {1 - e^{-\lambda \pi r_{\gamma}^2}} \bigg(  \frac{r_{\gamma}}{v} \int_0^2 e^{-\lambda r_{\gamma}^2 q(z)} dz - \frac{2e^{-\lambda \pi r_{\gamma}^2}r_{\gamma}}{v} \bigg).
\end{split}
\raisetag{1\baselineskip}
\end{equation}
The second equality is by the change of variables $vt=x$ and the last equality is by the change of variables $z = uv/r_{\gamma}$ and $q(z) = \frac{1}{2} z \sqrt{4 - z^2} + 2{\arctan}(\frac{z}{\sqrt{4 - z^2}}) $.

Now, from Equation \eqref{eq:integral} we have,
$$
\lim_{r_{\gamma} \to \infty} \mathbb{E}[U^{({\gamma})}]  \sim \frac{r_{\gamma}}{v} \int_0^2 e^{-\lambda r_{\gamma}^2 q(z)} dz.
$$
Now, the integral $\frac{r_{\gamma}}{v} \int_0^2 e^{-\lambda r_{\gamma}^2 q(z)}$ and $\frac{r_{\gamma}}{v} \int_0^2 e^{-\lambda r_{\gamma}^2 2z} dz$ are asymptotically equal as $r_{\gamma}$ goes to $\infty$.

The asymptotic equality of the two functions as $r_{\gamma}$ goes to $\infty$ means, that the relative error of the approximate equality goes to 0 as $r_{\gamma}$ goes to $\infty$ i.e.,

$$
\lim_{r_{\gamma} \to \infty} \frac{\int_0^2 e^{-\lambda r_{\gamma}^2 q(z)} - e^{-\lambda r_{\gamma}^2 2z} dz}{\int_0^2 e^{-\lambda r_{\gamma}^2 2z}} = 0.
$$

With the help of the Taylor series expansion, we get that
$$
\lim_{r_{\gamma} \to \infty} \frac{\int_0^{C/r_{\gamma}^2} e^{-\lambda r_{\gamma}^2 q(z)} - e^{-\lambda r_{\gamma}^2 2z} dz}{\int_0^2 e^{-\lambda r_{\gamma}^2 2z}} = 0,
$$
for some constant $C$.

Thus, we need to prove that 
$$ \lim_{r_{\gamma} \to \infty} \frac{\int_{C/r_{\gamma}^2}^2 e^{-\lambda r_{\gamma}^2 q(z)} - e^{-\lambda r_{\gamma}^2 2z} dz}{\int_0^2 e^{-\lambda r_{\gamma}^2 2z}} = 0.$$

Now let us observe the function $h(z) =  e^{-\lambda r_{\gamma}^2 q(z)} - e^{-\lambda r_{\gamma}^2 2z}$. The derivative of the function $h(z)$ is 

\begin{equation}
\label{equation}
\begin{split}
h'(z) & = \big(-\sqrt{4-z^2} e^{-\lambda r_{\gamma}^2 q(z)} + 2e^{-\lambda r_{\gamma}^2 2z}\big)r_{\gamma}^2 \\ \\ 
& = e^{-\lambda r_{\gamma}^2 2z}r_{\gamma}^2 \bigg[2 - \sqrt{4-z^2} e^{-\lambda r^2 [q(z)-2z]}\bigg].
\end{split}
\end{equation}

Thus by substituting $z = C/r_{\gamma}^2$ and using the Taylor series expansion of $[q(z) - 2z]$  and  $\sqrt{4-z^2}$ , we get
$$
\lim_{r_{\gamma} \to \infty} h'(C/r_{\gamma}^2) = e^{-2C\lambda z} [ o(r_{\gamma})] = 0.
$$

The derivative of function $h(z)$ is 0 for $z = C/r_{\gamma}^2$ as  $r_{\gamma}$ goes to $\infty$ i.e., the function $h(z)$ has a local maximum at $C/r_{\gamma}^2$ for some constant $C$ and it can be observed that it is decreasing for $C = 3$.

Therefore, the function $h(z)$ is a decreasing function from $3/r_{\gamma}^2$ to 2 and the area under its curve has an upper bound that goes to 0.

Thus, the two integrals are asymptotically equal as $r_{\gamma}$ goes to $\infty$, which implies that $\mathbb{E}[U^{({\gamma})}] \sim \frac{1}{2vr_{\gamma}\lambda}$ as $r_{\gamma}$ goes to $\infty$.

\end{document}